\documentclass[ejsv2]{imsart}

\usepackage{adjustbox}
\usepackage{algorithm}%
\usepackage{algorithmicx}%
\usepackage{algpseudocode}%
\usepackage{amsmath}
\usepackage[title]{appendix}%
\usepackage[british]{babel}
\usepackage{bigints}
\usepackage{bm}
\usepackage{booktabs}
\usepackage{float}
\usepackage{graphicx}
\usepackage{listings}%
\usepackage{mathtools}
\usepackage{multirow}%
\usepackage{siunitx}
\usepackage[dvipsnames]{xcolor}
\RequirePackage[numbers]{natbib}
\RequirePackage[colorlinks,citecolor=blue,urlcolor=blue]{hyperref}%Must be added last
\arxiv{2010.00000}
\startlocaldefs
\theoremstyle{plain}

\newtheorem{theorem}{Theorem}[section]
\newtheorem{lemma}[theorem]{Lemma}
\theoremstyle{definition}
\newtheorem{definition}[theorem]{Definition}

\theoremstyle{remark}

\endlocaldefs
\newcommand{\balpha}{\bm{\alpha}}

\newcommand{\bCD}{\mathcal{CD}}
\newcommand{\bD}{\mathcal{D}}
\newcommand{\bmu}{\bm{\mu}}
\newcommand{\bv}{\bm{v}}
\newcommand{\bX}{\bm{X}}

\newcommand{\bxo}{\bm{x}_{o}}

\newcommand{\bxoi}{\bm{x}_{o,i}}
\newcommand{\bx}{\bm{x}}
\newcommand{\by}{\bm{y}}

\newcommand{\bz}{\bm{z}}
\newcommand{\evg}{\tilde{E}_{ikg}^{(r)}}
\newcommand{\cvg}{\tilde{C}_{ikg}^{(r)}}
\newcommand{\mathR}{\mathbb{R}}
\newcommand{\ms}{\mathcal{M}}
\newcommand{\os}{\mathcal{O}}
\newcommand{\simplex}{\mathbb{S}_p}
\newcommand{\bpsi}{\bm{\psi}}

\definecolor{cerulean}{HTML}{2A52BE}
\begin{document}
\begin{frontmatter}
\title{Handling mild outliers and unobserved values in compositional datasets using finite mixtures of mean-parametrised Dirichlet models}
%\title{A sample article title with some additional note\thanksref{t1}}
\runtitle{Finite mixtures of mean-parametrised Dirichlet models}
%\thankstext{T1}{A sample additional note to the title.}

\begin{aug}
%%%%%%%%%%%%%%%%%%%%%%%%%%%%%%%%%%%%%%%%%%%%%%%
%% Only one address is permitted per author. %%
%% Only division, organization and e-mail is %%
%% included in the address.                  %%
%% Additional information can be included in %%

\author[A]{\fnms{Jason}~\snm{Pillay}\ead[label=e1]{js.pillay@up.ac.za}},
\author[A,B]{\fnms{Andriëtte}~\snm{Bekker}\ead[label=e2]{andriette.bekker@up.ac.za}\orcid{0000-0000-0000-0000}},
\author[C]{\fnms{Cristina}~\snm{Tortora}\ead[label=e3]{cristina.tortora@sjsu.edu}\orcid{0000-0001-8351-3730}},
and
\author[D]{\fnms{Antonio}~\snm{Punzo}\ead[label=e4]{antonio.punzo@unict.it}\orcid{0000-0001-7742-1821}}
%%%%%%%%%%%%%%%%%%%%%%%%%%%%%%%%%%%%%%%%%%%%%%
%% Addresses                                %%
%%%%%%%%%%%%%%%%%%%%%%%%%%%%%%%%%%%%%%%%%%%%%%
\address[A]{Department of Statistics,
University of Pretoria\printead[presep={,\ }]{e1,e2}}
%\runauthor{F. Author et al.}
\address[B]{
National Institute for Theoretical and Computational Sciences (NITheCS), Pretoria Node, University of Pretoria, Pretoria, 0028\printead[presep={,\ }]{e2}}
%\runauthor{F. Author et al.}
\address[C]{Department of Mathematics and Statistics,
San José State University\printead[presep={,\ }]{e3}}
%\runauthor{F. Author et al.}
\address[D]{Department of Economics and Business,
University of Catania\printead[presep={,\ }]{e4}}
%\runauthor{F. Author et al.}
\end{aug}

\begin{abstract}
%The abstract should summarize the contents of the paper.
%It should be clear, descriptive, self-explanatory and not longer
%than 200 words. It should also be suitable for publication in
%abstracting services. Formulas should be used as sparingly as
%possible within the abstract. The abstract should not make
%reference to results, bibliography or formulas in the body
%of the paper---it should be self-contained.
%
%This is a sample input file.  Comparing it with the output it
%generates can show you how to produce a simple document of
%your own.
Heterogeneous compositional data may be simultaneously affected by missing values and atypical points, posing challenges for both clustering and outlier detection. We develop a mixture model for incomplete compositional data under Huber’s contamination model, with contamination defined directly on the simplex and on observations that may be missing at random. The model provides a principled representation of outliers and allows the distribution of missing parts to be derived while accounting for contamination. We establish that maximisation of the observed log-likelihood constructed from contaminated, mean-parametrised Dirichlet densities is a convex optimisation problem. We then develop a tailored expectation-maximisation. The E-step incorporates the moments from the distribution of the missing parts of the data. Although the resulting parameter estimates are not available in closed form, the maximisation step admits tractable element-wise iterative updates. Numerical experiments demonstrate the performance of the proposed approach under varying percentage of missingness and contamination, and different sample size. An application to the American Time Use Survey identifies two interpretable clusters corresponding to work-intensive and sociable recreational days, while revealing atypical time-use compositions. In contrast, a conventional Dirichlet mixture model identifies four clusters, reflecting the influence of outliers and an artificial splitting of one cluster.

\end{abstract}



\end{frontmatter}
%%%%%%%%%%%%%%%%%%%%%%%%%%%%%%%%%%%%%%%%%%%%%%
%% Please use \tableofcontents for articles %%
%% with 50 pages and more                   %%
%%%%%%%%%%%%%%%%%%%%%%%%%%%%%%%%%%%%%%%%%%%%%%
%\tableofcontents

\section{Introduction}
For a given dimension $p \in \mathbb{N}_{+}$, a random compositional vector $\bX$ exists on the $p$-dimensional open unit simplex,
    \[\mathbb{V}_p 
    = \left\{ \bx = [x_1,\ldots,x_p]^{\top} 
    : \text{$x_k \in (0,1)$, $k \in \{1,\ldots,p\}$, $\sum_{k=1}^p x_k < 1$} \right\},\] and on the $p$-dimensional closed unit simplex
\begin{equation*}
    \mathbb{S}_p 
    = \left\{ \bx = [x_1,\dots,x_p]^{\top} 
    : \text{$x_k \in (0,1)$, $k \in \{1,\ldots,p\}$, $\sum_{k=1}^p x_k = 1$}
     \right\}.
\end{equation*}
The Dirichlet distribution is a multivariate generalisation of the Beta distribution, defining a family of unit sum-constrained probabilities or proportions on $\simplex$. A random vector $\bX\in\mathbb{S}_p$ is said to follow a Dirichlet distribution with parameter vector $\balpha = [\alpha_1,\ldots,\alpha_p]^{\top} \in \mathbb{R}_+^p$, denoted as $\bX \sim \mathcal{D}(\balpha)$, if its probability density function (density) is  %\textcolor{magenta}{We need to use a letter for this like  $f_{\mathcal{D}_\alpha}$}

\begin{equation}
\label{pdf_dirichlet}
    f_{\bD }(\bx;\balpha) 
    = \begin{cases}
        \dfrac{\Gamma(\alpha_0)}
        {\displaystyle\prod_{k=1}^p \Gamma(\alpha_k)}
        \displaystyle\prod_{k=1}^p x_k^{\alpha_k-1} &  \bx \in \mathbb{S}_p \\
        0         & \text{otherwise},
    \end{cases}
\end{equation}
where $\alpha_0 = \|\balpha\|_1 = \displaystyle\sum_{k=1}^p \alpha_k$, and $\Gamma(\cdot)$ denotes the gamma function. We refer to the density in \eqref{pdf_dirichlet} as the $\balpha$-parameterised Dirichlet distribution. This distribution is usually the first choice for analysing compositional data or measurements of proportions as it is supported directly on the simplex and has its moments in closed form with respect to the Lebesgue measure. It has shown applicability, especially as building blocks for increasingly complex models for various fields within statistics such as time series, Bayesian inference, and regression models, among others \cite{dirichlet_forensic, dirichlet_generator, dirichlet_kernel, dirichlet_matrices, dirichlet_neural_network, dirichlet_process,dirichlet_robotics,dirichlet_spatial_regression, dirichlet_time_series}.

Further, the Dirichlet has been extended to model complex distributions on the simplex. There are generalisations, extensions, and reparameterisations to enhance the Dirichlet's flexibility \citep{dirichlet_flexible,dirichlet_generalised, dirichlet_hyper, general_louiville}. In particular, we focus on modelling approaches that address the accommodation of outliers on the simplex. Model-based concepts around outliers stem from Huber's parametric contamination model \citep{huber_contamination}. Formally, for distribution $\bm P_{\bm \theta} $ with parameter set $\bm \theta$, and probability $\varepsilon \in (0,1)$, there exists a distribution $\bm P_{\bm \theta^*}$ belonging to the same family as $\bm P_{\bm \theta}$, but with a corresponding contaminated parameter set $\bm \theta^* \ne \bm \theta$ so that random vector $\bm X$ is drawn with probability
\[
\varepsilon \bm P_{\bm \theta^*} + (1 - \varepsilon)\bm P_{\theta}.
\]
This model has been applied to the Dirichlet distributions such as \citep{mode_cont_dirichlet, van2026contaminated}, under a mode and mean reparameterisation, respectively. 

The utility of the Dirichlet distribution is vast and widespread, especially in contexts of atypical points, but cannot be fitted to incomplete datasets, currently. Instead, other model-based approaches employ transformations to convert compositions from the simplex onto the real space. Currently, the multivariate normal distribution is assumed after incomplete datapoints are transformed using the additive log-ratio function (ALR) and then treated as normally distributed \citep{alr_chem_data, alr_eda, alr_em,alr_package}. Imputations and model-fitting are done using the normal distribution, but the estimated mean vector and covariance matrix do not correspond to estimating the mean vector and covariance matrix of the compositions themselves. Further, the transform depends on the choice of denominator, which introduces bias in imputations when back-transformed onto the simplex.

In approaches with complete data, a relied-upon transformation is the isometric log-ratio (ILR) due to its isometry \citep{ilr_analysis, ilr_kmeans, irl_regression}. Despite this, the positions of outliers relative to the bulk of the data are distorted. Thus, while transforming the data onto the real space in a way that preserves distance still poses a challenge for model-based techniques to detect them. In other words, to simultaneously account for missing values and outliers, there is a reliance on transformations to move the challenges into a space where techniques are available. This may be due, in part, to the fact that the very lack of techniques defined directly on the simplex forces the need for transformations. 
\\
\\
The rest of this paper is structured as follows: in Section \ref{prelims} we propose a finite mixture of Dirichlet distributions under parameter contamination. We change the $\balpha$ parameterisation to the mean-variability parameterisation as in \cite{van2026contaminated} so that the interpretation under Huber's model is consistent. We derive the distribution of the missing parts of a Dirichlet distributed random vector under a missing-at-random (MAR) mechanism in Section \ref{methodology}. We further prove that uniform noise on the simplex follows a Dirichlet distribution and use this fact as a benchmark for outlier detection. We then show in Section \ref{inference}, under the MAR assumption, the maximum-likelihood estimates (MLEs) of a contaminated mean-parametrised Dirichlet distribution exist and are unique and note under which conditions these apply for finite mixtures. We then develop a new algorithm to calculate the MLEs, and consequently provide cluster identification and outlier detection simultaneously. Performance of the algorithm is assessed in the simulation study in Section \ref{simulation}. The applicability is demonstrated in Section \ref{application}, where a 2-component contaminated mean-parametrised Dirichlet model is identified as the most appropriate fit for the dataset, containing more than 16\% of outliers despite more than 49\% of cells being missing.

%This sample helps you to create a properly formatted \LaTeXe\ manuscript.
%%%%%%%%%%%%%%%%%%%%%%%%%%%%%%%%%%%%%%%%%%%%%%
%% `\ ' is used here because TeX ignores    %%
%% spaces after text commands.              %%
%%%%%%%%%%%%%%%%%%%%%%%%%%%%%%%%%%%%%%%%%%%%%%
%Prepare your paper in the same style as used in this sample .pdf file. Try to avoid excessive use of italics and bold face.
%Please do not use any \LaTeXe\ or \TeX\ commands that affect the layout or formatting of your document (i.e., commands like \verb|\textheight|, \verb|\textwidth|, etc.).

\section{Preliminaries}
\label{prelims}
We briefly revisit some definitions and theorems needed for the novelty in the inferential procedure in Section \ref{inference}. Parameter estimation is seen as an optimisation problem, there is an overlap between calculus tools and statistical concepts. 

\subsection{Concepts in optimisation}
\begin{definition} [coerciveness]
\label{coercive_def}
Let $f : \mathR^p \rightarrow \mathR$ be a continuous function. Then $f$ is called coercive if
\begin{align}
\label{coercive_eq}
    \underset{||\bx||\rightarrow\infty}{lim}f(\bx) \rightarrow \infty.
\end{align}
\end{definition}

Definition \ref{coercive_def} leads to the following property that is useful in optimisation theory:
\begin{theorem}[Existence of optimiser]
\label{existence_theorem}
Let $f : \mathR^p \rightarrow \mathR$ be a continuous and coercive function and let $\mathbb{D}_f \subseteq \mathR^p$  be a nonempty, closed set. Then function $f$ has a global minimum point over $\mathbb{D}_f$. 
\end{theorem}

\begin{definition}[Local Lipschitz continuity]
    A function $f: \mathR^p \rightarrow \mathR$ is locally Lipschitz continuous at $\bx$ if there exists a $\delta_{\bx}>0$ and an $M_{\bx} >0$ so that
    \begin{align}
      ||f(\bx) - f(\by)|| \leq M_{\bx} ||\bx - \by||,
    \end{align}
    for all $\by$ in some neighbourhood of $\bx$. That is, for all $\by$ such that $||\bx - \by||<\delta_{\bx}$.
\end{definition}

\subsection{mean-parametrised Dirichlet distribution} 
A $p \times 1$ random vector $\bX$ follows a mean-parametrised Dirichlet distribution with mean vector $\bmu \in \simplex$ and variability parameter $\gamma >0$ if its probability density function (or density, for short) is given as:
\begin{align}
    \label{mean_dirichlet}
    f_{\mathcal{D}}(\bx;\bmu,\gamma) = \begin{cases}
        \frac{\Gamma\left(\frac{1}{\gamma} \right)}{\displaystyle\prod_{k=1}^p\Gamma\left(\frac{\mu_k}{\gamma} \right)} \displaystyle\prod_{k=1}^px_k^{\frac{\mu_k}{\gamma}-1} & \bx \in \simplex \\
        0 & \text{otherwise},
    \end{cases}
\end{align}
and is denoted as $\bX \sim \mathcal{D}_{\simplex}(\bmu,\gamma)$. The density in \eqref{mean_dirichlet} is related to the $\balpha$-parameterised Dirichlet distribution through the following links:
\begin{align*}
    \alpha_k{ \gamma } = \mu_k         &\implies \mu_k = \frac{\alpha_k}{\displaystyle ||\balpha||_1} \\
    ||\balpha||_1 = \frac{1}{\gamma} &\implies  \gamma = \frac{1}{||\balpha||_1}.
\end{align*}

For the $\balpha$ parameterisation, the variance of the $k^{th}$ element of $\bX$, namely $X_k$ is :
\begin{align*}
    \sigma_k^2 = var(X_k) = \frac{\frac{\alpha_k}{||\balpha||_1}\left(1 - \frac{\alpha_k}{||\balpha||_1}\right)}{||\balpha||_1 +1 } = \frac{\mu_k\left(1 - \mu_k \right)}{1 +\frac{1}{\gamma} }.
\end{align*}

Now, for a fixed value of $\mu_k$, the first derivative of $\sigma^2_k$ with respect to $\gamma$ is:
\begin{align*}
    \frac{\partial\sigma_k^2}{\partial\gamma} = \frac{ \mu_k(1-\mu_k ) }{ \left(1 + \frac{1}{\gamma}\right)^2 \gamma^2} = \frac{\frac{\sigma_k^2}{\gamma}}{ \left(1 + \frac{1}{\gamma}\right)},
\end{align*}
which is strictly positive for all values of $\gamma>0$. Thus, an increase in $\gamma$ is sufficient to cause an increase in $\sigma_k^2$. Further, this is not restricted to a specific value of $k$, implying that $\gamma$ increases the variance for all $k = 1, \dots,p$. Parameter $\gamma$ thus controls the variation in the distribution. The concept of accommodating outlier points through an inflated variation can be applied to density \eqref{mean_dirichlet} to produce the contaminated mean-parametrised Dirichlet distribution.
\subsection{Contaminated mean-parametrised Dirichlet distribution} 
A $p \times 1$ random vector $\bX$ follows a contaminated mean-parametrised Dirichlet distribution with mean vector $\bmu \in \simplex$, a variability parameter $\gamma >0$, a variation inflation parameter $\eta>1$, and a parameter $\varepsilon \in (0,1)$ if its density is given as:
\begin{align}
    \label{cont_dirichlet}
    f_{\mathcal{CD}}(\bx; \bmu,\gamma, \eta, \varepsilon) = \begin{cases} \underset{\text{contaminated component}}{\underbrace{\varepsilon f_{\mathcal{D}}(\bx;\bmu,\eta \gamma)}} + \underset{\text{reference component}}{\underbrace{(1-\varepsilon)f_{\mathcal{D}}(\bx;\bmu, \gamma)}}& \bx \in \simplex \\
        0 & \text{otherwise},
    \end{cases}
\end{align}
and is denoted as $\bX \sim \mathcal{CD}_p(\bmu,\gamma, \eta, \varepsilon)$. Under Huber's model, the parameter epsilon can be interpreted as the proportion of the population that are outliers, provided that $\varepsilon \in (0,0.5)$. That is
%\textcolor{magenta}{This is only true if $\epsilon<0.5$, Antonio and I disagree on this point. To avoid conflict we can add a sentence saying that if we add the constraint on $\epsilon$, the inflated component can be interpreted as modelling outlier} from the majority, that is
\begin{align*}
    \varepsilon = \mathbb{P}(\bX \text{ is an outlier point}).
\end{align*}
%CITE OCKERT'S MASTERS
Notice that $f_{\mathcal{CD}}(\bx;\bmu,\gamma) \rightarrow f_{\mathcal{D}}(\bx;\bmu,\gamma)$ as $\eta \rightarrow 1$ and/or as $\varepsilon \rightarrow 0$. In other words, the reference mean-parametrised Dirichlet density is a limiting case of its contaminated generalisation. For brevity, the density in \eqref{cont_dirichlet} and corresponding distribution will be referred to simply as the contaminated Dirichlet density and distribution, respectively.

\subsection{Finite mixtures of contaminated Dirichlet distributions}
When the observed proportions arise from multiple underlying subpopulations, a finite mixture of contaminated Dirichlet distributions provides a flexible model. 
The mixture density is given by
\begin{align}
    f_{\mathcal{CDM}}(\bx;\bpsi) = \sum_{g=1}^G \pi_g f_{\mathcal{CD}}(\bx;\bmu_g,\gamma_g, \eta_g, \varepsilon_g),
    \label{fmm}
\end{align}
where $\bpsi = \{\pi_g,\bmu_g,\gamma_g, \eta_g,\varepsilon_g: g=1,\dots,G\}$. 
The mixing proportions satisfy $\pi_g>0$, $\displaystyle\sum_{g=1}^G \pi_g = 1$, and $f_{\mathcal{CD}}(\cdot)$ is defined in \eqref{cont_dirichlet}.

\section{Methodology}
\label{methodology}
This section introduces the assumptions behind missing values, from which the distributions of the observed and missing parts of the random vector are established. Hereafter, a random vector $\bX \sim \mathcal{CD}(\bmu,\gamma,\eta,\varepsilon)$ is partitioned according to its observed and missing subvectors, as follows:
\[ 
\bX =
\begin{bmatrix}
\bX_m \\
\bX_o
\end{bmatrix},
\]
where the subscripts $m$ and $o$ denote the missing and observed parts, respectively. Let $\mathcal{M} \subset \{1,\dots,p\}$ denote the index set of missing parts, with cardinality $p_m = |\mathcal{M}|$, and let $\mathcal{O} = \{1,\dots,p\} \setminus \mathcal{M}$ denote the index set of observed parts, the complement set of $\mathcal{M}$, with cardinality $p_o = |\mathcal{O}|$. Accordingly, $\bX_m$ and $\bX_o$ denote the subvectors of $\bX$ formed by the parts indexed by $\mathcal{M}$ and $\mathcal{O}$, respectively.

\subsection{Mechanisms behind missing values}
As in any incomplete-data setting, a composition is informative for inference when at least one component is observed. 
However, compositional vectors exist in the simplex and are bounded. Whenever exactly one component is unobserved, its value is uniquely determined by the observed ones, following from the constraint $\displaystyle\sum_{k=1}^p X_k = 1$. In other words, this scenario has one equation and one unknown variable. Consequently, non-identifiability arises only when at least two components are unobserved, a feature that is unique to compositional data. This also implies that the dimension for non-trivial inference of incomplete compositions must be higher than 2.

In practice, components may be unobserved for a variety of reasons, such as non-response or measurement limitations. Accordingly, we assume a missing-at-random (MAR) mechanism, together with its special case, namely missing-completely-at-random (MCAR), as a realistic and flexible approach for modelling the missingness process. Under the standard ignorability conditions for likelihood-based inference, this formulation accommodates a broad range of mechanisms giving rise to unobserved values while allowing inference to proceed without explicitly modelling the missingness mechanism. 
\subsection{\texorpdfstring{Distributions of $\bX_m$ and $\bX_o$}{Distributions of Xm and Xo}  }
In particular, under the MAR assumption, the conditional distribution of the missing parts, $\bX_m$ given the observed ones $\bX_o$ admits a closed-form expression, given in the following theorem:

\begin{theorem}
    Let $\bX \sim \bD_{\mathbb{S}_p}(\balpha)$, and partition $\bX$ and $\bmu$ according to the index sets $\mathcal{M}$ and $\mathcal{O}$ as
$
\bX =
\begin{bmatrix}
\bX_m \\
\bX_o
\end{bmatrix}
\text{ and }
\bmu =
\begin{bmatrix}
\bmu_m \\
\bmu_o
\end{bmatrix}.
$
Let $p_m = |\mathcal{M}|$, $p_o = |\mathcal{O}|$, and letting $c(\bxo)=1-\|\bxo\|_1>0$ we have:
\begin{equation*}
\label{cond}
\frac{\bX_m}{c(\bxo)}
\;\bigg|\;
\bX_o=\bxo
\sim
\mathcal{D}(\bmu^*, \gamma^*),
\end{equation*}
where 
$\bmu^* = \frac{\bmu_m}{c(\bmu_o)}$ and $\gamma^* = \frac{\gamma}{c(\mu_o)}$.
\end{theorem}
%%%%%%%%%%%%%
\begin{proof}
First, recognise that $\frac{\bm{X}_m}{c(\bxo)} \sim \bD (\balpha_m)$. The proof of this result can be found in Theorem 2.5 of \cite[Section 2.2]{ng2011dirichlet}. All that is left is to recognise the change of variables, as required.
\end{proof}

Further, the marginal distribution of the observed parts is also tractable, given in the following theorem
\begin{theorem}
\label{marginal_dirichlet}
Let $\bX \sim \bD_{\simplex}(\bmu,\gamma)$, and partition $\bX$ and $\bmu$ according to the index sets $\mathcal{M}$ and $\mathcal{O}$ as
\[
\bX =
\begin{bmatrix}
\bX_m \\
\bX_o
\end{bmatrix}
\text{ and } 
\bmu =
\begin{bmatrix}
\bmu_m \\
\bmu_o
\end{bmatrix}.
\]
Define 
\[\bmu^v_o = \begin{bmatrix}
\bmu_o \\
 1 - \| \bmu_o\|_1
\end{bmatrix}\quad \text{and} \quad 
\bX^v_o = \begin{bmatrix}
\bX_o \\
 1 - \| \bX_o\|_1
\end{bmatrix}.\]
Then the marginal density of the observed component $\bX^v_o$ is given by
\begin{equation*}
\bX_o^v
\sim
\mathcal{D}_{\mathbb{S}_{p_o+1}}(\bmu^v_o, \gamma).
\end{equation*}
\end{theorem}
% Consider noise on Sp to perofrm a sensitivity analysis....motivate
\subsection{\texorpdfstring{Noise on $\simplex$}{Noise on simplex}}
%\textcolor{magenta}{When using the CN we do not assume the noise is uniform, that would be a gross outlier. We need to talk about this.}
The concept of an outlier on the simplex is different from Euclidean space. In the latter space, an outlier is characterised by its proximity to the majority of the data. However, the simplex is bounded. Thus, outliers are relative to the positions of the clusters and the simplex. Consequently, the notions around noise on Euclidean space must be appropriated for the simplex. A suitable starting point would be to consider what noise on the simplex looks like. On the real space, it is often observations generated from some uniform distribution. However, observations from a uniform distribution on the Euclidean space do not necessarily imply they are uniformly distributed on the simplex.

Contaminating observations with uniform noise provides a useful benchmark for assessing the sensitivity of contaminant models. Such contamination can produce observations that range from mild outliers, which lie relatively close to the underlying clusters, to gross outliers located far from the bulk of the data, while some contaminated observations may remain hidden within the clusters themselves. Thus, uniform contamination provides a general scenario in which multiple forms of atypical observations are present simultaneously, allowing the performance and sensitivity of contaminant models to be evaluated under varying types of contamination.

\citep{spacings} shows that uniformly distributed observations on the simplex can be created from uniform observations on the Euclidean space, through their spacings, defined as follows:
\begin{definition}[Spacings]
\label{spacing_def}
    Suppose $U_1,\dots,U_{p-1}$ are independent observations from a uniform distribution on the unit interval, denoted as $U_k \sim unif(0,1)$. Denote the order statistics as $U^{(0)}\leq U^{(1)}\leq,\dots, \leq U^{(p-1)} \leq U^{(p)}$, where $U^{(0)} = 0$ and $U^{(p)}=1$. The elements of the vector $\bm S$ are the successive differences between the order statistics. That is, $S_k = U^{(k)} - U^{(k-1)}$ for $k=1,\dots,p$.
\end{definition}
From Definition \ref{spacing_def}, $\bm S \in \simplex$. The next theorem then shows that $\bm S$ is uniformly distributed on $\simplex$
\begin{theorem}
    The spacings, $\bm S$, from a sample of $p-1$ independent observations from a $unif(0,1)$ distribution are uniformly distributed on $\simplex$.
    \begin{proof}
        It is sufficient to see that the spacings are dependent on the ordering of the sample from a $unif(0,1)$ distribution. Thus, its density is determined on the number of ways $U_1,\dots,U_p$ can be ordered, which is $(p-1)!$ ways (see \citep{spacings}). Thus, the density of $\bm S$ is given as 
        \begin{align}
            f(\bm S ) = \begin{cases} (p-1)! & \bm S \in \simplex \\ 0 & \mathrm{otherwise.}  \end{cases}
        \end{align}
    \end{proof}
\end{theorem}
Lastly, we contextualise the distribution of $\bm S$ within the distributions of the simplex:
\begin{theorem}
   Let $\bm S$ be a vector of spacings from a uniformly distributed random sample of size $p-1$. Then $\bm S \sim \bD_{\simplex}\left(\frac{1}{p}\bm 1; \frac{1}{p} \right)$, where $\bm 1$ is a $p\times 1$ vector whose elements are one.
    \begin{proof}
        For $\bm S \in \simplex$, its density is:
        \begin{align*}
            (p-1)! = \Gamma\left( p \right) = \frac{\Gamma\left( \frac{1}{\frac{1}{p} } \right)}{\displaystyle\prod_{k=1}^p\Gamma \left( \frac{\frac{1}{p} }{ \frac{1}{p}}\right) }\prod_{k=1}^p s_k^{\frac{\frac{1}{p} }{\frac{1}{p}} - 1},
        \end{align*}
        which is the density of a $ \bD_{\simplex}\left(\frac{1}{p}\bm 1; \frac{1}{p} \right)$ distribution.
    \end{proof}
\end{theorem}

\section{Inference}
\label{inference}
Consider a random sample $\mathcal{X} = \{ \bx_1,\dots, \bx_n\}$ generated by \eqref{fmm}, whose vectors are potentially incomplete. That is, each observation $\bx_i$ can be partitioned into missing and observed subvectors,
$
\bx_i =
\begin{bmatrix}
\bx_{i,m}\\
\bx_{i,o}
\end{bmatrix},
$
where the subscripts $m$ and $o$ refer to the missing and observed components, respectively. Under the MAR mechanism, the maximum likelihood estimator (MLE) $\hat{\bpsi}$ is the optimiser of the observed log-likelihood, given as
\begin{align}
\label{observed_ll}
    \ell_o(\bpsi; \mathcal{X}_o) = \sum_{i=1}^n \ln \sum_{g=1}^G \pi_g f_{\mathcal{CD}}(\bx_{i,o}^v; \bmu_{i,og}^v, \gamma_g, \eta_g, \varepsilon_g),
\end{align}
where $\mathcal{X}_o = \{ \bxoi: i=1,\dots,n \}$. In preparation for section \ref{estimation}, we establish existence and uniqueness of the estimates for the case $G=1$, and note the restrictions sufficient for model identifiability, i.e. under which circumstances does existence and uniqueness apply for $G>1$.
\subsection{Existence and uniqueness}
Under the case $G=1$, \eqref{observed_ll} simplifies to 
\begin{align}
\label{observed_ll_1}
    \ell_o(\bpsi; \mathcal{X}_o) = \sum_{i=1}^n \ln f_{\mathcal{CD}}(\bx_{i,o}^v;\bmu_{i,o}^v, \gamma, \eta, \varepsilon).
\end{align}
Here, $\ell_o: \bm{\Psi} \rightarrow \mathR$ is defined on the parameter space $\bm{\Psi}$, which is the Cartesian product $\simplex^{\epsilon} \times \mathR_+^{\epsilon} \times [1+ \epsilon,\infty) \times [0+\epsilon,1-\epsilon]$ where $\simplex^{\epsilon}=\{\bmu \in \simplex: \mu_k \geq \epsilon,k=1,\dots,p\}$, and $\mathR_+^{\epsilon} = \{\gamma \in \mathR: \gamma\geq \epsilon\}$, for some fixed $\epsilon>0$. 

To prove existence, we must show that the bounded parts of $\Psi$, namely $\simplex^{\epsilon} \times [0+\epsilon,1-\epsilon]$ are compact, and ensure that $\ell_o$ does not increase arbitrarily as $\gamma$ and $\eta$ increases. Now, since $\simplex$ is a convex set, $\simplex^{\epsilon}$ is a closed and bounded subset of $\simplex$. Thus, $\simplex^{\epsilon} \times [0+\epsilon,1-\epsilon]$ is a compact set.

For the case $ \begin{bmatrix} \gamma \\  \eta \end{bmatrix} \rightarrow \begin{bmatrix} \epsilon \\ 1+ \epsilon \end{bmatrix}$, $\ell_o$ is defined and bounded. All that is left, is to consider the unbounded cases of  $ \begin{bmatrix} \gamma \\  \eta \end{bmatrix}$. This is handled by the following lemma.
\begin{lemma}
\label{coercive_lemma} $-\ell_o(\bpsi; \mathcal{X}_o)$ is coercive on $\bm{\Psi}$.
\begin{proof} See Appendix \ref{coercive_proof} \end{proof}
\end{lemma}
Lemma \ref{coercive_lemma} permits the use of Theorem \ref{existence_theorem}, to conclude that there exist $\hat{\bmu}, \hat{\gamma}, \hat{\eta},$ and $\hat{\varepsilon}$ that are optimisers of $\ell_o$. To prove uniqueness, we aim to show that \eqref{observed_ll_1} is strictly concave on $\bm{\Psi}$. The strategy employed is to first demonstrate blockwise concavity. We begin with $\bmu$, while fixing $\gamma,\eta,\varepsilon$.
Consider
\begin{align}
\frac{\partial f_{\mathcal D}(\bx_{i,o}^v;\bmu_{i,o},\gamma) }{\partial\mu_k}
&=  -\frac{ \frac1\gamma\Gamma\!\left(\frac1\gamma\right) }{\displaystyle\prod_{j\in\os_i} \Gamma\!\left(\frac{\mu_j}{\gamma}\right)}
\displaystyle \prod_{j\in\os_i} x_{ij}^{\frac{\mu_j}{\gamma}-1}
 \frac{\partial}{\partial\mu_k} \Gamma\!\left(\frac{\mu_k}{\gamma}\right)
+\frac{ \frac1\gamma\Gamma\!\left(\frac1\gamma\right)} {\displaystyle\prod_{j\in\os_i} \Gamma\!\left(\frac{\mu_j}{\gamma}\right)}
\displaystyle \prod_{j\in\os_i} x_{ij}^{\frac{\mu_j}{\gamma}-1}\ln x_{ik} \nonumber
\\
&= \frac{\Gamma\!\left(\frac1\gamma\right)}{ \displaystyle\prod_{j\in\os_i} \Gamma\!\left(\frac{\mu_j}{\gamma}\right)}
\displaystyle\prod_{j\in\os_i} x_{ij}^{\frac{\mu_j}{\gamma}-1} 
\left( -\frac{ \frac{ \partial}{ \partial_{\mu_k} }  \Gamma \left(\frac{\mu_k}{\gamma} \right)} {\gamma\Gamma(\frac{\mu_k}{\gamma})} +\frac{\ln x_{ik}}{\gamma}
\right)\nonumber
\\
% &= \frac{\Gamma\!\left(\frac1\gamma\right)} {\displaystyle\prod_{j\in\os_i}\Gamma\!\left(\frac{\mu_j}{\gamma}\right)}
% \displaystyle\prod_{j\in\os_i} x_{ij}^{\frac{\mu_j}{\gamma}-1} \left( -\psi\!\left(\frac{\mu_k}{\gamma}\right) +\ln x_{ik}^{\frac1\gamma}
% \right)\nonumber\\
&= - f_{\mathcal D}(\bx_{i,o}^v;\bmu_{i,o}^v,\gamma)\delta(\mu_k), \nonumber
\end{align}
where $\delta(\mu_k) = \varphi\!\left(\frac{\mu_k}{\gamma}\right) -\frac{1}{\gamma}\ln x_{ik}$, with $\varphi(\cdot)$ denoting the digamma function. The Hessian, $\bm H_{i,o}$ of $f_{\mathcal D}(\cdot)$ with respect to $\bmu$ can be written as:
\[
%\frac{\partial^2 f_{\mathcal D}(\bx_{i,o}^v;\bmu_{i,o},\gamma)}{\partial\mu_\ell\partial\mu_k}
\bm H^{(lk)}_{i,o}
=
\begin{cases}
-\dfrac{\partial f_{\mathcal D}(\bx_{i,o}^v;\bmu_{i,o}^v,\gamma)}{\partial\mu_\ell}\delta(\mu_k),
&
\ell\neq k,
\\[1.2em]
-\dfrac{\partial f_{\mathcal D}(\bx_{i,o}^v;\bmu_{i,o}^v,\gamma)}{\partial\mu_\ell}\delta(\mu_k) -
\dfrac{f_{\mathcal D}(\bx_{i,o}^v;\bmu_{i,o}^v,\gamma) }\gamma \varphi^{(1)} \!\left( \dfrac{\mu_k}{\gamma} \right),
&
\ell=k,
\end{cases}
\]
where $\varphi^{(1)}(\cdot)$ denotes the trigamma function. With respect to $\bmu$, if we can show that the Hessian matrix of $\ell_o$, say $\bm \nabla^2 \ell_o$, is negative definite, it can be concluded that $\ell_o$ is strictly concave. Since $\ell_o$ is a sum of terms involving $f_{\mathcal D}(\cdot)$, we can show each $\bm H_{i,o}$ is negative definite. 

\begin{lemma}
\label{neg_def_lemma}
    For observed sample $\mathcal{\bm X}_o$ of size $n \in \mathbb N$, the matrix $\bm H _{i,o}$ is negative definite $\forall~i=1,\dots,n$.
\end{lemma}
\begin{proof}
    See Appendix \ref{neg_def_proof}.
\end{proof}
Using Lemma \ref{neg_def_lemma} we can determine the strict concavity of $\ell_o$ and thus conclude a unique maximiser $\hat{\bmu}$ exists, through the following theorem.
\begin{theorem}
\label{unique_mu}
    For fixed $\gamma, \eta$ and $\varepsilon$, the function $g^{(1)}(\bmu; \gamma, \eta,\varepsilon, \mathcal {\bm X}_o) = \ell_o(\bm\Psi;\mathcal{\bm X}_o)$ as a function of $\bmu$, has a unique maximiser, say $\hat{\bmu}$, where $\ell_o$ is as given in \eqref{likelihood_fmm}, provided that the observed sample spans the parameter space. 
    \begin{proof}
        From Lemma \ref{neg_def_lemma}, $f_\mathcal D(\cdot)$ is strictly concave on $\simplex^\epsilon$ with respect to $\bmu$. $f_{\bD}(\cdot)$ is strictly log-concave on $\simplex^\epsilon$, and $\ell_o$ is concave on $\simplex^\epsilon$. Now, define $\bm O_i$ to be the projection of $\bmu$ onto its observed space, i.e. $\bm O_i \bmu = \bmu_{i,o}$. Then, for any $ \bz\neq0$, 
        \[ \bz^\top \nabla^2\ell_o \bz = \sum_{i=1}^{n} (\bm O_i\bz) ^\top  \bm H_{i,o}(\bm O_i \bz )= 0 \] 
        will only be possible if, $\forall~\bmu\in\simplex^\epsilon$ the observed sample fails to span the parameter space. Therefore
        \[ \bz^\top \bm \nabla^2\ell_o \bz \prec \bm 0. \]
Since a maximiser of $\ell_o$, namely $\hat{\bmu}$ exists, and function  $g^{(1)}(\bmu; \gamma, \eta,\varepsilon, \mathcal {\bm X}_o)$ is strictly concave on $\simplex^\epsilon$, $\hat{\bmu}$ is unique.
    \end{proof}
\end{theorem}
Theorem \ref{unique_mu} implies that any maximisers $\hat{\gamma}, \hat{\eta},$ and $\hat{\varepsilon}$ thus share $\hat \bmu$. This means $\ell_o$ will have maximisers that lie on the surface $g^{(1)}(\hat \bmu; \gamma, \eta,\varepsilon, \mathcal {\bm X}_o)$. 

Now $\ell^o$ is twice-smooth, hence focus on\[g^{(2)}_i( \gamma; \eta, \varepsilon, \bx_{i,o}, \hat \bmu) = \ell_{o,i} \left( \bx_{i,o}; \hat \bmu, \gamma, \eta \right)
= \sum_{i=1}^n\ln\!\left[ (1-\varepsilon)e^{\phi_1} + \varepsilon e^{\phi_2} \right],\]
where
\[
\phi_1 = \ln\Gamma\!\left(\frac{1}{\gamma}\right) - \sum_{k\in\os_i} \ln\Gamma\!\left(\frac{\hat\mu_k}{\gamma}\right) + \sum_{k\in\os_i} \ln x_{ik}
\left( \frac{\hat\mu_k}{\gamma}-1 \right), \]

and

\[
\phi_2= \ln\Gamma\!\left(\frac{1}{\eta\gamma}\right) - \sum_{k\in\os_i}\ln\Gamma\!\left( \frac{\hat\mu_k}{\eta\gamma} \right) + \sum_{k\in\os_i} \ln x_{ik}\left( \frac{\hat\mu_k}{\eta\gamma}-1 \right).
\]

Then \[ 
\dfrac{\partial g^{(2)}_i}{\partial \gamma} = \frac{ (1-\varepsilon)e^{\phi_1}}{f_{\bCD} \left( \bx_{i,o}; \hat\bmu, \gamma, \eta \right) } \frac{\partial\phi_1}{\partial\gamma} + \frac{\varepsilon e^{\phi_2} }{
f_{\bCD} \left( \bx_{i,o}; \hat\bmu, \gamma, \eta \right) } \frac{\partial\phi_2}{\partial\gamma}.
\]
Let $t=\frac{1}{\gamma}$. The sign of $\dfrac{\partial g^{(2)}_i}{\partial \gamma}$ depends on $\frac{\partial\phi_1}{\partial\gamma}$ and $\frac{\partial\phi_2}{\partial\gamma},\forall\ t > 0$ where,
\[
\frac{\partial^2\phi_1}{\partial t^2}
= \varphi^{(1)}(t) - \sum_{k\in\os_i} \hat\mu_k^2 \varphi^{(1)}(\hat\mu_k t).
\]
and
\[
\frac{\partial^2\phi_2}{\partial t^2} = \frac{1}{\eta^2} \left[ \varphi^{(1)}\!\left(\frac{t}{\eta}\right) - \sum_{k\in\os_i} \hat\mu_k^2 \varphi^{(1)}\!\left(\frac{\hat\mu_k t}{\eta} \right) \right].
\]
We now show that $\frac{\partial\phi_1}{\partial\gamma}$ and $\frac{\partial\phi_2}{\partial\gamma}$ are strictly decreasing. We need the following lemma:
\begin{lemma}
    The function $h(x) = x\varphi^{(1)}(x)$ is strictly decreasing on $x>0$.
    \begin{proof}
    The series expansion of $\varphi^{(1)}$ is \citep{trigamma}:
    \begin{align*}
        \varphi^{(1)}(x) = \sum_{j=0}^{\infty} \frac{1}{(j + x)^2 } \implies h(x) =  \sum_{j=0}^{\infty} \frac{x}{(j + x)^2 } 
    \end{align*}
    Observe that for $y = x + h$ where $h>0$, so that $x<y$: 
    \begin{align*}
        h(y) - h(x) &= \sum_{j=0}^{\infty} \left[ \frac{y}{(j + y)^2}  - \frac{x}{(j + x)^2} \right] \\
                    &= \sum_{j=0}^{\infty} \left[ \frac{x^2y - xy^2}{(x+j)^2(y+j)^2} \right]\\
                    &= \sum_{j=0}^{\infty} \left[ \frac{x^2(x+h) - x(x+h)^2}{(x+j)^2(y+j)^2} \right]\\
                    &= \sum_{j=0}^{\infty} \left[ - \frac{xh^2}{(x+j)^2(y+j)^2} \right]\\
                    &< 0 ~\forall h,x>0.
    \end{align*}
    \end{proof}
\end{lemma}

Thus, for any $t>0$ we have:
\begin{align*}
    \mu_k t < t &\implies t \varphi^{(1)}(t) < \mu_k t \varphi^{(1)}(\mu_kt)\\
                &\implies \varphi^{(1)}(t) < \mu_k \varphi^{(1)}(\mu_kt).\\
                \text{ Multiplying by $\mu_k>0$ on both sides}\\
\mu_k \varphi^{(1)}(t) < \mu_k^2 \varphi^{(1)}(\mu_kt) &\implies \sum_{k\in \os_i} \mu_k \varphi^{(1)}(t) < \sum_{k\in \os_i}  \mu_k^2 \varphi^{(1)}(\mu_kt).
\end{align*}
Recall from Theorem \ref{marginal_dirichlet}, the mean, $\bm \mu^v_o$, is an element of $\mathbb{S}_{p_o}$. Thus, $ \displaystyle\sum_{k\in \os_i} \mu_k =1$, $\frac{\partial\phi_1}{\partial\gamma}$ and $\frac{\partial\phi_2}{\partial\gamma}$ are strictly decreasing on $(0,\infty)$. In other words, $\dfrac{\partial g^{(2)}_i}{\partial \gamma}$ changes sign at most once. Hence if a maximiser for $\gamma$ exists, it is unique. This brings us to the final theorem of the section.
\begin{theorem}
    \label{all_unique}
    The observed log-likelihood function $\ell_o$ defined on $\bm \Psi$ has maximisers $\hat{\bmu},\hat{\gamma},$ $\hat{\eta}$, and $\hat \varepsilon$ and these maximisers are unique.
    \begin{proof}
        From Theorem \ref{existence_theorem}, maximisers of $\ell_o$ exist. By Theorem \ref{unique_mu}, $\hat{\bmu}$ exists, for fixed $\gamma, \eta$, and $\varepsilon$. Since $\hat{\gamma}$ exists, and the derivative of $\ell_o$ with respect to $\gamma$ is strictly decreasing, we have that $\hat{\gamma}$ is unique. Notice that, inductively, $\hat \eta$ exists and is unique by analogous argument to $\hat \gamma$. Lastly, it is shown that the derivative of the observed log-likelihood with respect to $\varepsilon$ is strictly decreasing, implying at most one solution exists, say $\hat \varepsilon$. Combined with Theorem \ref{existence_theorem}, we thus have $\hat \varepsilon$ is unique.
    \end{proof}
\end{theorem}
Lastly, we consider the case for $G>1$. Uniqueness of a solution cannot be guaranteed in general. \citep{identifiability} shows that if $G<p$, finite mixtures of density \eqref{mean_dirichlet} are identifiable, up to permutation. %Existence of a maximiser can still be guaranteed, since a 

\subsection{EM-based algorithm}
\label{estimation}
A random sample $\mathcal{X}$ that is generated by \eqref{fmm} is incomplete for three reasons: (1) some parts of the row may be missing, (2) it is not known which of the $G$ components the observation is generated from, and (3) it is not known whether said row is an outlier point or not. There is a latent cluster membership vector $\bz_i = \begin{bmatrix} z_{i1},\dots, z_{iG}\end{bmatrix}^{\top}$ where $z_{ig} = 1$ if $\bx_i$ is generated from the $g^{th}$ component of model \eqref{fmm} and 0 otherwise, subject to the constraint $||\bz_i||_1 =1 $, for $i=1,\dots,n$. Similarly, a latent vector $\bv_{ig} = \begin{bmatrix} v_{i1},\dots, v_{iG}\end{bmatrix}^{\top}$ such that $v_{ig}$ assumes the value 1 when $\bx_i$ is an outlier point and 0 otherwise, conditioned on whether $\bx_i$ belongs to the $g^{th}$ component. Let $\ms_i \subset \{1,\dots,p\}$ and $\os_i$ denote the index set of missing and observed parts of $\bx_i$, respectively. Letting $\mathcal{Z} = \{\bz_i: i =1,\dots,n \}$ and $\mathcal{V}= \{\bv_i: i =1,\dots,n \}$, the complete likelihood is therefore given as

\begin{align} 
\tilde{\mathcal{L}}(\bpsi; \mathcal{X}, \mathcal{V}, \mathcal{Z} ) 
= \prod_{i=1}^n \prod_{g=1}^G &\left\{ \pi_g 
\left[\varepsilon_g f_{\bD}(\bx_i;\bmu_g, \eta_g\gamma_g) \right]^{v_{ig}} 
\left[(1-\varepsilon_g) f_{\bD}(\bx_i;\bmu_g, \gamma_g) \right]^{1-v_{ig}} \right\}^{z_{ig}} 
 \nonumber\\
= \prod_{i=1}^n \prod_{g=1}^G & \left\{ 
\pi_g
\left[ \varepsilon_g \frac{\Gamma\left( \frac{1}{\eta_g\gamma_g} \right) }{\displaystyle\prod_{k=1}^p \Gamma\left(\frac{\mu_{kg}}{\eta_g\gamma_g}\right) }
\displaystyle\prod_{k \in \ms_i} x_k^{\frac{\mu_{kg}}{\eta_g\gamma_g} -1} \displaystyle\prod_{k \in \os_i} x_k^{\frac{\mu_{kg}}{\eta_g \gamma_g} -1} \right]^{v_{ig}} \right. \nonumber\\
&\qquad\left.
\left[(1 - \varepsilon_g) \frac{\Gamma\left( \frac{1}{\gamma_g} \right) }{\displaystyle\prod_{k=1}^p\Gamma\left(\frac{\mu_{kg}}{\gamma_g}\right) } \displaystyle\prod_{k \in \ms_i} x_k^{ \frac{\mu_{kg}}{\gamma_g} -1} \displaystyle\prod_{k \in \os_i} x_k^{\frac{\mu_{kg}}{\gamma_g} -1} \right] ^{1-v_{ig}} \right\} ^{z_{ig}}. 
\label{likelihood_fmm} 
\end{align}
\subsubsection{E-step}
Letting $\tilde \ell _c$ denote the natural log of \eqref{likelihood_fmm}, the E-step computes $\tilde{Q}(\bpsi ) = \mathbb{E}\left[\tilde \ell _c(\bpsi; \mathcal{X}, \mathcal{V}, \mathcal{Z} )\big|\bpsi^{(r)}\right]$, using the parameter updates at the $r^{th}$ iteration, namely $\bm{\bpsi}^{(r)}$, $\tilde{Q}(\bpsi )$ depends on the following expectations:
\begin{align}
    z^{(r)}_{ig} &:= \mathbb{E}\left[ Z_{ig}        | \bm \psi^{(r)}, \bx_{i,o} \right], \quad
    v^{(r)}_{ig} := \mathbb{E}\left[ V_{ig}         | \bm \psi^{(r)}, \bx_{i,o}, z_{ig} = 1\right] , \label{latent} \\  
    \evg         &:= \mathbb{E}\left[ \ln \bX_{i,m} | \bm \psi^{(r)}, \bx_{i,o}, v_{ig}=0, z_{ig} =1\right] ,\quad \text{and } \label{missing_ref} \\
    \cvg         &:= \mathbb{E}\left[ \ln \bX_{i,m} | \bm \psi^{(r)}, \bx_{i,o}, v_{ig}=1, z_{ig} =1\right] \label{missing_cont}
\end{align}

The expressions for the expectations in \eqref{latent} are:
\begin{align*}
    z^{(r)}_{ig} = \mathbb{P}(Z_{ig} =1|\bm\psi^{(r)}, \bx_{i,o}) &= \frac{\pi_g f_{\mathcal{CD}}\left( \bx_{i,o}^v; \bmu_{i,o,g}^{v~(r)}, \gamma_g^{(r)} \eta_g^{(r)}, \varepsilon_g^{(r)} \right) }
    { \displaystyle\sum_{j=1}^G \pi_jf_{\mathcal{CD}}\left( \bx_{i,o}^v; \bmu_{i,o,j}^{v~(r)}, \gamma_j^{(r)} \eta_j^{(r)}, \varepsilon_j^{(r)} \right)}\\
    v^{(r)}_{ig} = \mathbb{P}(V_{ig} =1|\bm\psi^{(r)}, \bx_{i,o}, z_{ig} = 1) &= 
    \frac{ \varepsilon^{(r)}_gf_{\mathcal{D}}\left( \bx_{i,o}^v; \bmu_{i,o,g}^{v~(r)}, \eta_g^{(r)}\gamma_g^{(r)} \right) }
    {  f_{\mathcal{CD}}\left( \bx_{i,o}^v; \bmu_{i,o,g}^{v~(r)}, \gamma_g^{(r)}, \eta_g^{(r)}, \varepsilon_g^{(r)} \right)}
\end{align*}
By Theorem \ref{cond}, the expected values in \eqref{missing_ref} and \eqref{missing_cont} have the following expression:
\begin{align*}
        \evg &=  \ln \left( 1 - ||\bx_{i,o}||_1 \right)  +  \varphi\left(\frac{\mu_{kg}^{(r)}}{\gamma^{(r)}_g }  \right) - \varphi\left(\frac{ \left|\left|\bmu_{i,m,g}^{(r)}\right|\right|_1 } {  \gamma^{(r)}_g } \right),\\
        \cvg &=  \ln \left( 1 - ||\bx_{i,o}||_1 \right)  +  \varphi\left(\frac{\mu_{kg}^{(r)}}{\gamma^{(r)}_g\eta^{(r)}_g }  \right) - \varphi\left(\frac{ \left|\left|\bmu_{i,m,g}^{(r)}\right|\right|_1 } { \gamma^{(r)}_g \eta^{(r)}_g } \right),
\end{align*}
where $ \varphi(\cdot)$ denotes the digamma function.
Thus, $\tilde{Q}(\bpsi )$ is given as:
\begin{align}
\label{q_psi}
    \tilde{Q}(\bpsi ) & = \sum_{i=1}^n \sum_{g=1}^G z ^{(r)}_{ig}\ln\pi^{(r)}_g + \sum_{i=1}^n \sum_{g=1}^G z^{(r)}_{ig}v^{(r)}_{ig}\ln\varepsilon^{(r)}_g + \sum_{i=1}^n \sum_{g=1}^G z^{(r)}_{ig}\left(1-v^{(r)}_{ig}\right)\ln\left(1-\varepsilon^{(r)}_g\right) \nonumber\\
    &+\sum_{i=1}^n \sum_{g=1}^G z^{(r)}_{ig}v^{(r)}_{ig} \ln\Gamma\left( \frac{1}{\eta^{(r)}_g\gamma^{(r)}_g} \right)
    +\sum_{i=1}^n \sum_{g=1}^G z^{(r)}_{ig} \left(1-v^{(r)}_{ig}\right) \ln\Gamma\left( \frac{1}{\gamma^{(r)}_g} \right)\nonumber\\
    &- \sum_{i=1}^n \sum_{g=1}^G \sum_{k=1}^p z^{(r)}_{ig}v^{(r)}_{ig}\ln \Gamma\left(\frac{\mu^{(r)}_{kg}}{\eta^{(r)}_g\gamma^{(r)}_g}\right)
    - \sum_{i=1}^n \sum_{g=1}^G \sum_{k=1}^p z_{ig}(1-v_{ig})\ln \Gamma\left(\frac{\mu_{kg}}{\gamma_g}\right) \nonumber\\
    &+\sum_{i=1}^n \sum_{g=1}^G \sum_{k \in \ms_i} z^{(r)}_{ig}v^{(r)}_{ig}\frac{\mu^{(r)}_{kg}}{\eta^{(r)}_g\gamma^{(r)}_g}\cvg
    +\sum_{i=1}^n \sum_{g=1}^G \sum_{k \in \ms_i} z^{(r)}_{ig}\left(1-v^{(r)}_{ig}\right)\frac{\mu^{(r)}_{kg}}{\gamma^{(r)}_g}\evg\nonumber\\
    &+\sum_{i=1}^n \sum_{g=1}^G \sum_{k \in \os_i} z^{(r)}_{ig}v^{(r)}_{ig}\frac{\mu^{(r)}_{kg}}{\eta^{(r)}_g\gamma^{(r)}_g}\ln x_{ik}
    +\sum_{i=1}^n \sum_{g=1}^G \sum_{k \in \os_i} z^{(r)}_{ig}\left(1-v^{(r)}_{ig}\right)\frac{\mu^{(r)}_{kg}}{\gamma^{(r)}_g}\ln x_{ik}\nonumber\\
     &- \sum_{i=1}^n \sum_{g=1}^G \sum_{k \in \ms_i} z^{(r)}_{ig}v^{(r)}_{ig}\cvg
    - \sum_{i=1}^n \sum_{g=1}^G \sum_{k \in \ms_i} z^{(r)}_{ig}\left(1-v^{(r)}_{ig}\right)\evg\nonumber\\
    &- \sum_{i=1}^n \sum_{g=1}^G \sum_{k \in \os_i} z^{(r)}_{ig}v^{(r)}_{ig}\ln x_{ik}
    - \sum_{i=1}^n \sum_{g=1}^G \sum_{k \in \os_i} z^{(r)}_{ig}\left(1-v^{(r)}_{ig}\right)\ln x_{ik}.
\end{align}

%and $\beta(\cdot)$ is the variability, dependent on whether the observation is generated by the reference or contaminated component:
%\begin{align*}
%    \beta\left(v_{ig}; \gamma^{(r)}, \eta^{(r)} \right) = 
%    \begin{cases}
%        \gamma^{(r)}     & \text{when } v_{ig} =0,\\
%        \eta^{(r)}\gamma^{(r)} & \text{when } v_{ig} =1.
%    \end{cases}
%`\end{align*}
\subsubsection{CM-Step}
The M-step at the $(r+1)^{th}$ iteration maximises \eqref{q_psi} with respect to $\bm\Psi$. Maximising with respect to $\pi_g$, and $\varepsilon$, subjected to the constraint mentioned in \eqref{fmm} updates $\pi_g^{(r+1)}$ as
\begin{align*}
    \pi_g^{(r+1)} = \frac{ n_g }{n}, \quad \varepsilon_g^{(r+1)} = \frac{n_{\mathrm{cont.},g }}{n_g}
\end{align*}
where $ n_g = \displaystyle\sum_{i=1}^n z^{(r)}_{ig}$ and $n_{\mathrm{cont.},g} =\displaystyle\sum_{i=1}^n z^{(r)}_{ig}v^{(r)}_{ig} $, for $g=1,\dots,G$. Maximising with respect to $\bmu_g$, $\gamma_g$, and $\eta_g$ is complicated for two reasons: (1) differentiating \eqref{q_psi} does not provide tractable estimators, and (2) $\bmu_g$ is subjected to the constraint $\| \bmu_g\|_1 =1$ for $g=1,\dots,G.$ For these reasons, $\bmu_g$, $\gamma_g$, and $\eta_g$ are estimated numerically through the constrained Newton-Raphson method. Conditioned on $\gamma_g$ and $\eta_g$, the partial derivative of \eqref{q_psi} with respect to $\bmu_{kg}$ is:
\begin{align*}
\dfrac{\partial \ell(\bmu_{kg})}{\partial\bmu_{kg}} &=  \frac{n_{\mathrm{ref.},g}}{\gamma^{(r)}_g} \varphi\!\left(\frac{\mu_{kg}}{\gamma^{(r)}_g}\right) 
+ \frac{1}{\gamma^{(r)}_g}\sum_{i=1}^{n}z^{(r)}_{ig}\left(1-v^{(r)}_{ig}\right)\left[\sum_{k \in \ms_i} \ln \evg 
+\sum_{k \in \os_i} \ln x_{ik} \right] \\
&- \frac{n_{\mathrm{cont.},g} }{\eta^{(r)}_g\gamma^{(r)}_g} \varphi\!\left(\frac{\mu_{kg}}{\eta^{(r)}_g\gamma^{(r)}_g}\right) + \frac{1}{\eta^{(r)}_g\gamma^{(r)}_g} \sum_{i=1}^{n}z^{(r)}_{ig} v^{(r)}_{ig} \left[\sum_{k \in \ms_i} \ln \cvg + \sum_{k \in \os_i}\ln x_{ik}\right],
\end{align*}
with $n_{\mathrm{ref.},g} = n_g - n_{\mathrm{cont.},g}$, so that the gradient with respect to $\bmu_g$ evaluated at some $\bmu^{(s)}_g$ is:
\begin{align*}
    \bm \nabla_{\bmu_g}^{(s)}\tilde{Q}
    %& = \begin{bmatrix}\dfrac{\partial \ell(\bmu_{1g})}{\partial\bmu_{1g}},&\dots&,& \dfrac{\partial \ell(\bmu_{pg})}{\partial\bmu_{pg}}\end{bmatrix} ^ \top.\nonumber\\
    & = \begin{bmatrix}\nabla^{(s)}_{\bmu_{1g}}\tilde{Q},&\dots&,& \nabla^{(s)}_{\bmu_{pg}}\tilde{Q} \end{bmatrix} ^ \top.
\end{align*}
The $kj^{th}$ element of the Hessian with respect to $\bmu_g$ is
\[
h_{kjg} =
\begin{cases} -\dfrac{1}{(\gamma_g^{(r)})^2} \varphi^{(1)}\!\left(\dfrac{\mu_{kg}}{\gamma_g^{(r)}}\right) n_{\mathrm{ref.},g}
- \dfrac{1}{(\eta_g^{(r)}\gamma_g^{(r)})^2} \varphi^{(1)}\!\left(\dfrac{\mu_{kg}}{\eta_g^{(r)}\gamma_g^{(r)}}\right) n_{\mathrm{cont.},g} & k=j,
\\
0 & k\neq j.
\end{cases}
\]
Consequently, when $h_{kjg}$ is evaluated at $\bmu^{(s)}_g$, it is denoted as $h^{(s)}_{kjg}$.
The nested update for $\bmu_g$ can be given elementwise as:
\begin{align*}
    \bmu_{kg}^{(s+1)} = \bmu_{kg}^{(s)} + \frac{\nabla^{(s)}_{\mu_{kg}} \tilde{Q} }{h^{(s)}_{kkg}} - \frac{\lambda^{(s)}_g }{h^{(s)}_{kkg}},
\end{align*}
where
\[
\lambda^{(s)}_g
=
 \frac{ \displaystyle \sum_{k=1}^p \frac{\nabla^{(s)}_{\mu_{kg}} \tilde{Q} }{h^{(s)}_{kkg}} }{ \displaystyle\sum_{k=1}^p \frac{1}{h^{(s)}_{kkg}}}.
\]
At convergence, the updates terminate and $\bmu_g^{(r+1)}$ is obtained. Now let \[
\gamma_g=e^{\omega_g},
\iff
\omega_g=\ln\gamma_g.
\]

Then the first two derivatives of \eqref{q_psi} with respect to $\omega_g$ evaluated at the $s^{th}$ update of $\omega_g$ is
\begin{align*}
    \nabla^{(s)}_{\omega_g}\tilde{Q} &= A^{(s)}_ge^{-\omega^{(s)}_g}, \text{ and }\\
    h^{(s)}_{\omega_g}               &= -A^{(s)}_ge^{-\omega^{(s)}_g} - B^{(s)}_ge^{-2\omega^{(s)}_g}.
\end{align*}
where,
\begin{align}
A^{(s)}_g =& \frac{n_{\mathrm{cont.,}g}}{\eta^{(r)}_g} \left[ \sum_{k=1}^{p} \mu^{(r+1)}_{kg} \varphi\!\left( \frac{\mu^{(r+1)}_{kg}}{\eta^{(r)}_g e^{\omega^{(s)}_g}} \right) 
- \varphi\!\left( \frac{1}{\eta^{(r)}_g e^{\omega^{(s)}_g}} \right) \right]
\nonumber\\&
+ n_{\mathrm{ref.,}g} \left[ \sum_{k=1}^{p} \mu^{(r+1)}_{kg} \varphi\!\left( \frac{\mu^{(r+1)}_{kg}}{e^{\omega^{(s)}_g}}\right)
- \varphi\!\left( \frac{1}{e^{\omega^{(s)}_g}} \right) \right]
\nonumber\\
&- \frac{1}{\eta^{(r)}_g} \left[ \sum_{i=1}^{n}\sum_{k\in\mathcal{M}_i} z^{(r)}_{ig}v^{(r)}_{ig}\cvg\mu^{(r+1)}_{kg} 
+ \sum_{i=1}^{n}\sum_{k\in\mathcal{O}_i} z^{(r)}_{ig}v^{(r)}_{ig}\mu^{(r+1)}_{kg}\ln x_{ik} \right]
\nonumber\\
&- \left[ \sum_{i=1}^{n}\sum_{k\in\mathcal{M}_i} z^{(r)}_{ig}(1-v^{(r)}_{ig})\evg\mu^{(r+1)}_{kg}
+ \sum_{i=1}^{n}\sum_{k\in\mathcal{O}_i} z^{(r)}_{ig}(1-v^{(r)}_{ig})\mu^{(r+1)}_{kg}\ln x_{ik}
\right],
\end{align}
and 
\begin{align*}
B^{(s)}_g =& \frac{n_{\mathrm{cont.,}g}}{\left(\eta^{(r)}_g\right)^2} \left[ \sum_{k=1}^{p}( \mu^{(r+1)}_{kg})^{2}
\varphi^{(1)}\!\left( \frac{\mu^{(r+1)}_{kg}}{\eta^{(r)}_g e^{\omega^{(s)}_g}} \right)
- \varphi^{(1)}\!\left( \frac{1}{\eta^{(r)}_g e^{\omega^{(s)}_g}} \right) \right]
\nonumber\\&
+ n_{\mathrm{ref.,}g} \left[ \sum_{k=1}^{p} (\mu^{(r+1)}_{kg})^{2} \varphi^{(1)}\!\left( \frac{\mu^{(r+1)}_{kg}}{e^{\omega^{(s)}_g}} \right)
- \varphi^{(1)}\!\left( \frac{1}{e^{\omega^{(s)}_g}} \right) \right].
\end{align*}
Thus the NR algorithm updates $\omega_g$ as:
\begin{align*}
    \omega_g^{(s+1)} = \omega_g^{(s)} + \frac{1}{1 + \frac{B^{(s)}_g}{A^{(s)}_g}e^{-\omega^{(s)}_g} },
\end{align*}
to which $\gamma^{(s+1)}_g$ is updated as $ \omega_g^{(s+1)}$ is at convergence. Lastly, let
\[
\eta_g=e^{\delta_g},
\iff
\eta_g=\ln\delta_g.
\]
and the first two derivatives with respect to $\delta$, when evaluated at $\delta^{(s)}$ are
\begin{align*}
    \nabla_{\delta_g}^{(s)}\tilde{Q} &=  C_g^{(s)} \frac{e^{-\delta_g^{(s)}}}{\gamma_g^{(r+1)}} \text{ and } \\
    h^{(s)}_{\delta_g}               &= -C_g^{(s)} \frac{e^{-\delta_g^{(s)}}}{\gamma_g^{(r+1)}} - D_g^{(s)}\frac{e^{-2\delta_g^{(s)}}}{\gamma_g^{(r+1)}}
\end{align*}
where,
\begin{align*}
C_g^{(s)} &= \left\{ n_{\mathrm{cont.,}g}\left[ \varphi\!  \sum_{k=1}^{p} \mu_{kg}^{(r+1)} \varphi\!\left( \frac{\mu_{kg}^{(r+1)} e^{-\delta_g^{(s)}}}{\gamma_g^{(r+1)}} \right) - \left( \frac{e^{-\delta_g^{(s)}}} {\gamma_g^{(r+1)}} \right) \right]
\right.\\ & \quad \left.
- \sum_{i=1}^{n}\sum_{k \in \ms_i}z^{(r)}_{ig}v^{(r)}_{ig}\mu_{kg}^{(r+1)}\cvg 
- \sum_{i=1}^{n}\sum_{k \in \ms_i}z^{(r)}_{ig}v^{(r)}_{ig}\mu_{kg}^{(r+1)}\ln x_{ik}\right\},
\end{align*}
 and 
\begin{align*}
D^{(s)}_{g}
&=  n_{\mathrm{cont.,}g} \left[ \varphi\!  \sum_{k=1}^{p} \mu_{kg}^{(r+1)} \varphi\!\left( \frac{\mu_{kg}^{(r+1)} e^{-\delta_g^{(s)}}}{\gamma_g^{(r+1)}} \right) - \left( \frac{e^{-\delta_g^{(s)}}} {\gamma_g^{(r+1)}} \right) \right]
\end{align*}

Thus the NR algorithm updates $\delta_g$ as:
\begin{align*}
    \delta_g^{(s+1)} = \delta_g^{(s)} + \frac{1}{ 1 + \frac{D^{(s)}_{g}}{C_g^{(s)}}e^{-\delta_g^{(s)}}},
\end{align*}
to which $\eta^{(s+1)}_g$ is updated as $ \delta_g^{(s+1)}$ is at convergence.

\subsection{Initialisation and convergence}
Initialising parameters for EM-based algorithms benefit from non-parametric clustering techniques, suck as k-means, k-medoids, hierarchical clustering, and so on. In fact, there are the simplex-based analogues of said initialisation techniques that specific for compositional data. Unfortunately, packages in \texttt{R} do not directly use incomplete compositions. Instead, the missing values would need to be imputed first and then clustered. From literature, a suitable procedure is to impute the incomplete values via $k$ nearest neighbours (knn), \citep{zCompositions} closed to sum to one, and then apply a non-parametric technique. In this paper, Aitchison k-means is used to provide initial cluster partitions to mimic how mixture models would be initialised on the Euclidean space \citep{kmeans_compositional}.\\ 

Within the EM-based algorithm, the NR algorithm requires suitable starting values in order to maximise the chances of convergence to the global maximum of the log-likelihood. Particular attention is paid to the initialisation of $\bmu_g$, which is solved for by a constrained NR algorithm. The formulation in the M-step implicitly ensures the updates remain positive. However, due to the additive constraint, the formulation is only valid when the initial value of $\bmu_g$ adheres to the constraint. There is a variant that allows for $\bmu_g$ to be initialised without necessarily satisfying the additive constraint, but for purposes of model-fitting, this is not needed, as it is preferred to initialise population moments with sample moments, which already satisfy the additive constraint. As a result, the method-of-moments estimates provide simple and reliable starting values for $\bmu_g$ and $\gamma_g$ \citep{estimation_NR}. The sample mean of the $g^{\text{th}}$ component and the variability of the sample are computed using the imputed observations as follows
\begin{align}
\overline{\bm x}_{g} = \frac{1}{n_g} \sum_{i=1}^{n_g} \bm x_{i,\mathrm{knn}}, \qquad s^2 =\frac{1}{n-1}\sum_{i=1}^n \sum_{k=1}^p (x_{ik,\mathrm{knn}} - \overline{x}_{kg})^2.
\end{align}
The parameters governing contamination are typically initialised to values that position the contaminated density \eqref{cont_dirichlet} to its reference density \eqref{mean_dirichlet}, with the conservative assumption that there are no outlier points, leaving the algorithm to report otherwise. Further, the existence and uniqueness properties established require compact parameter spaces for the bounded parameters. Thus, it is favourable to set $\eta_g$ amd $\varepsilon$ to their respective lower bounds 
\[ 1 + \epsilon \text{ and } \epsilon, \] 
where $\epsilon$ is a predetermined value, such as machine precision. In this paper, $\epsilon = 0.001$.\\

Two stopping rules must be specified for the proposed algorithm. 
An inner stopping rule controls the convergence of the Newton--Raphson (NR) procedure used within the M-step, while an outer stopping rule determines the convergence of the overall EM algorithm.

For the NR algorithm, convergence is monitored using the decrement $\frac{1}{2}D_g^2(\bm \theta^{(s+1)})$ for parameter $\bm \theta$ where:
\[ D_g  = \sqrt{ \left(\Delta \bm \theta_g^{(s+1)}  \right)^{\top} \bm H^{(s+1)} \Delta \bm\theta_g^{(s+1)}    },
\]
with $\Delta \bm \theta_g^{(s+1)} = \bm \theta_g ^{(s+1)} - \bm \theta_g ^{(s)}$ and $\bm H^{(s+1)}$ is the Hessian matrix evaluated at $\bm \theta_g ^{(s+1)}$. The criterion, $\frac{1}{2}D_g^2$ provides an approximation to the distance from the current objective function value to its supremum, namely $\left| Q(\bm{\theta}_g) - \underset{\bm\theta}{\mathrm{sup}}Q(\bm\theta)\right|$. Thus, iterations are stopped once $\frac{1}{2}D_g^2$  falls below a small tolerance when iterating over the estimator for the $g^{th}$ cluster.

The observed-data log-likelihood increases monotonically under the EM algorithm. However, temporary stability may occur when the algorithm approaches a local maximum before moving towards the global maximum. To assess convergence to the asymptotic log-likelihood, we employ the Aitken acceleration criterion. Let $l_o^{(r)}$ denote the observed log-likelihood at the $r^{\text{th}}$ EM iteration. The Aitken acceleration is defined as
\begin{align*}
a^{(r+1)} =
\frac{l_o^{(r+2)}-l_o^{(r+1)}}{l_o^{(r+1)}-l_o^{(r)}} .
\end{align*}
The corresponding estimate of the asymptotic log-likelihood is
\begin{align*}
(l_o^{\infty})^{(r)} =
l_o^{(r+1)} +
\frac{l_o^{(r+2)}-l_o^{(r+1)}}{1-a^{(r+1)}} .
\end{align*}
The EM algorithm is considered to have converged when
\[
(l_o^{\infty})^{(r)}-l_o^{(r+1)} < \epsilon,
\]
where $\epsilon>0$ is a small tolerance \citep{convergence_aitken}.

Finally, a subtle point worth mentioning about the NR algorithm is that, strictly speaking, numerically solves for the roots of the gradient. Thus, it is possible for the algorithm to converge to a stationary point that is not the maximum, or for the search direction to cause the algorithm to diverge. Theorem \ref{all_unique} ensures strict concavity, ruling out the former being a problem. Much like EM-based algorithms, convergence of the NR is not necessarily guaranteed. However, it can be supported that, should the starting value be sufficiently close, the algorithm will converge to the maximum. Formally, it is given in the following theorem
\begin{theorem}[Convergence of NR algorithm]
    Suppose that $\tilde{Q}(\bm \theta)$ is twice differentiable and that the Hessian is negative definite and Lipschitz continuous in a neighbourhood of a stationary point $\bm \theta^*$. Then, if the initial point $\bm \theta^0$ is sufficiently close to $\bm \theta^*$:
    \begin{itemize}
        \item[1.] The sequence of iterations converge to $\bm \theta^*$. 
        \item[2.] The iterations $\bm \theta^{(s)}$ converge quadratically to $\bm \theta^*$.
    \end{itemize}
    \begin{proof}
        See Theorem 3.7 of \citep{newton_convergence}.
    \end{proof}
\end{theorem}
It can be shown that the Hessian of $\tilde{Q}$ given in \eqref{q_psi} is locally Lipschitz continuous around its maximisers. Thus, the NR iterations converge quadratically to the maximiser.

\subsection{Clustering and outlier detection}
The construction of the complete-data log-likelihood \eqref{q_psi} allows one to determine cluster membership through their respective maximum a posteriori (MAP) probabilities. That is, setting $\hat{z}_{ig}$ to be the value of $z^{(r)}_{ig}$ at convergence of the EM-based algorithm, an observation $\bx_{i,o}$ is deemed part of the $g^{th}$ cluster if \begin{align}
\label{responsibilities}
    \hat{z}_i =\underset{g}{ \mathrm{argmax}}~\hat{z}_{ig}.
\end{align}

 The second indicator, $v_{ig}$ further flags typical and outlier points after all points are assigned a cluster. Thus, the EM-based algorithm offers outlier detection capability via the following a posteriori probability for an observation $\bx_{i,o}$, provided that $\varepsilon_g \in (0,0.5)$:
 \begin{align}
\label{is_outlier}
    \hat{v}_i = \mathbb{I}(\hat{v}_{ig}>0.5|\hat{z}_i),
\end{align}
 where $\mathbb{I}(\cdot)$ denotes the indicator function. For $\hat{v}_{i}>0.5,~ \bx_{i,o}$ is considered an outlier. 
 
 Conveniently, a posteriori probabilities $\hat{z}_i$ and $\hat{v}_i$ are automatically calculated as part of the EM-based algorithm's $k^{th}$ iteration.

\section{Simulation study}
\label{simulation}
The goal of the simulation experiment is to evaluate three aspects of the EM-based algorithm: (1) the clustering capability of incomplete compositions, (2) the model selection performance, and (3) the capability of algorithm to detect outlier points. In this study, 1000 datasets are generated as a mixture of $G=2$ Dirichlet distributions contaminated with noise, and subject to missingness at random. The choice of clusters is to reflect the results obtained in the application in Section \ref{application}. Similarly, the parameter choice is computed according to the degree of pairwise cluster overlap apparent in the data.

\begin{align*}
    \bmu_1 =    \begin{bmatrix}
                    \num{0.06555531}\\
                    \num{0.26965129}\\
                    \num{0.13295868}\\
                    \num{0.13295868}\\
                    \num{0.13295868}\\
                    \num{0.13295868}\\
                    \num{0.13295868}
                \end{bmatrix},
\qquad
\bmu_2 =        \begin{bmatrix}
                    \num{0.26965129}\\
                    \num{0.06555531}\\
                    \num{0.13295868}\\
                    \num{0.13295868}\\
                    \num{0.13295868}\\
                    \num{0.13295868}\\
                    \num{0.13295868}
                \end{bmatrix},
\qquad \gamma_1 = \gamma_2 =  \num{0.01889638}, 
\qquad 
\pi_1 = \num{0.65},~ \pi_2 = \num{0.35} 
\end{align*}    
The parameter $\bmu$ is chosen to have overlap over all but two elements, to mimic the overlap estimated in the dataset. The choice of $\gamma$ also corresponds to the level of cluster overlap as measured by the cluster overlap coefficient defined by \cite{ovc}.

The ability of the algorithm to accurately cluster incomplete observations from the reference distribution and incomplete observations and identify outlier points, despite their incompleteness are investigated. Lastly, the algorithm's ability to recover the true number of clusters present in the dataset via different model selection criteria is observed.
Throughout the study, the following scenarios are considered:
\begin{itemize}
    \item[(a).] The percentage of missing cells ranges from 0\% to 90\% in increments of 10\%.
    \item[(b).] The percentage of noise added ranges from 10\% to 30\% in increments of 10\%.
    \item[(c).] Small sample size ($n=100$) and large sample size ($n=1000$). 
\end{itemize}

\subsection{Results}
The average accuracy and ARI scores for the 1000 datasets are plotted in \figurename~\ref{fig:acc_ari}. As expected, the accuracy and ARI drop as the percentage of values in the dataset increases, with sample size becoming a stronger divisor between the performances. Interestingly, the percentage of noise has a greater effect on accurately clustering the typical points, with an apparent concave-like accuracy line for $n$=100 with 30\% of rows replaced with noise. Here, the full dataset has a worse accuracy compared to cases where values are missing. This can be explained by what noise on the simplex is. Combining a small sample with a relatively high percentage of outliers means that, out of 70 data points from the reference distribution, some random configurations of outlier points and reference points are such that the clusters appear different to what the underlying distribution is. This is motivated by the improvement in cluster performance as the percentage of missing values increases -- the effect of the noise on altering the clusters is lessened as there are fewer of these points. Ultimately, for higher percentages of missing values, the lack of information---outlier or otherwise---is enough to have a negative impact on accuracy and ARI scores.

\begin{figure}[H]
    \centering
    \includegraphics[trim={0cm 0.5cm 0cm 0.7cm}, clip, width=0.7\linewidth]{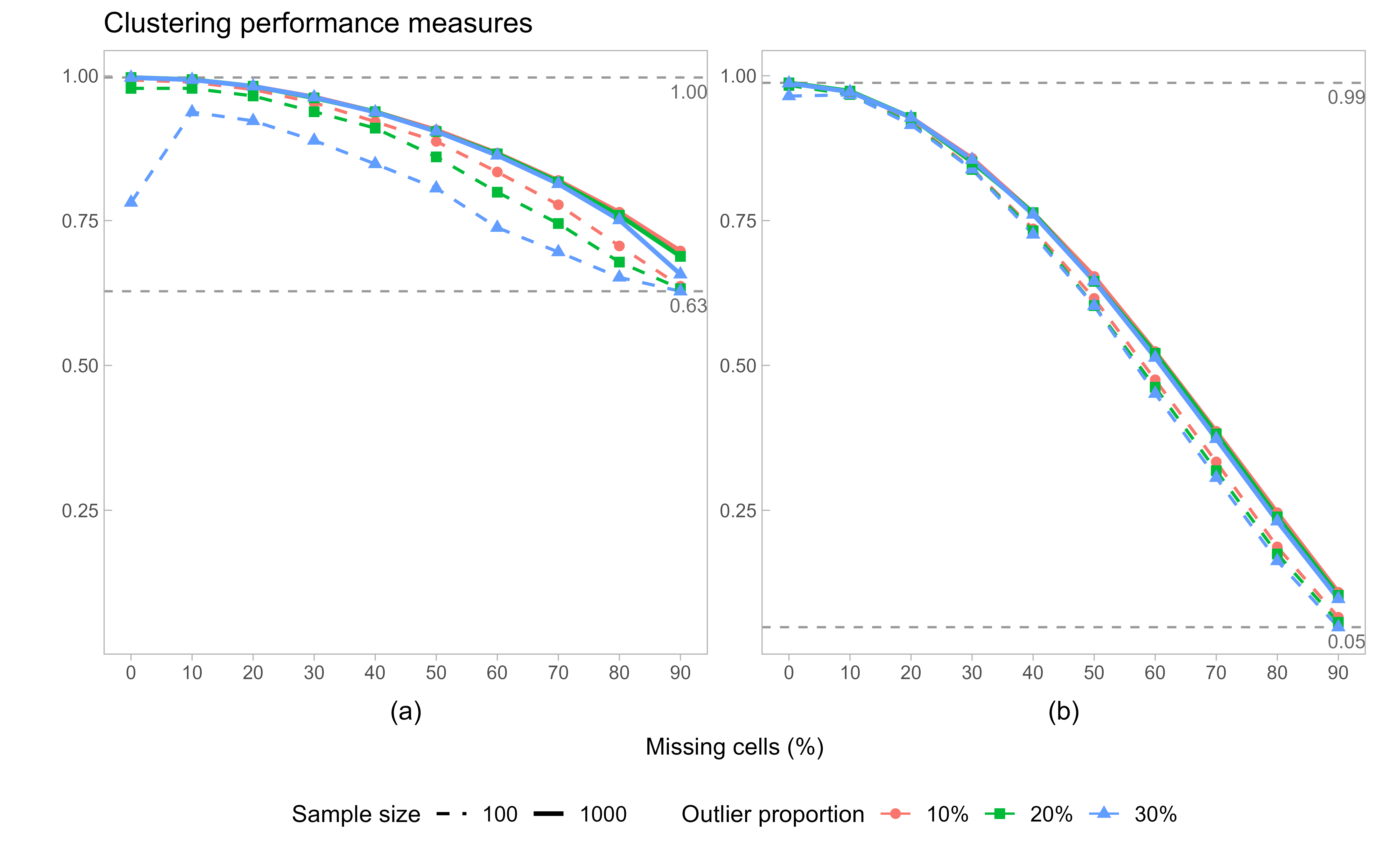}
    \caption{Cluster performance metrics: (a) Average accuracy scores and (b) average ARI scores.}
    \label{fig:acc_ari}
\end{figure}

%\subsection{Outlier detection performance}

\begin{figure}[H]
    \centering
    \includegraphics[trim={0.5cm 0.5cm 0cm 2.3cm}, clip, width=0.98\linewidth]{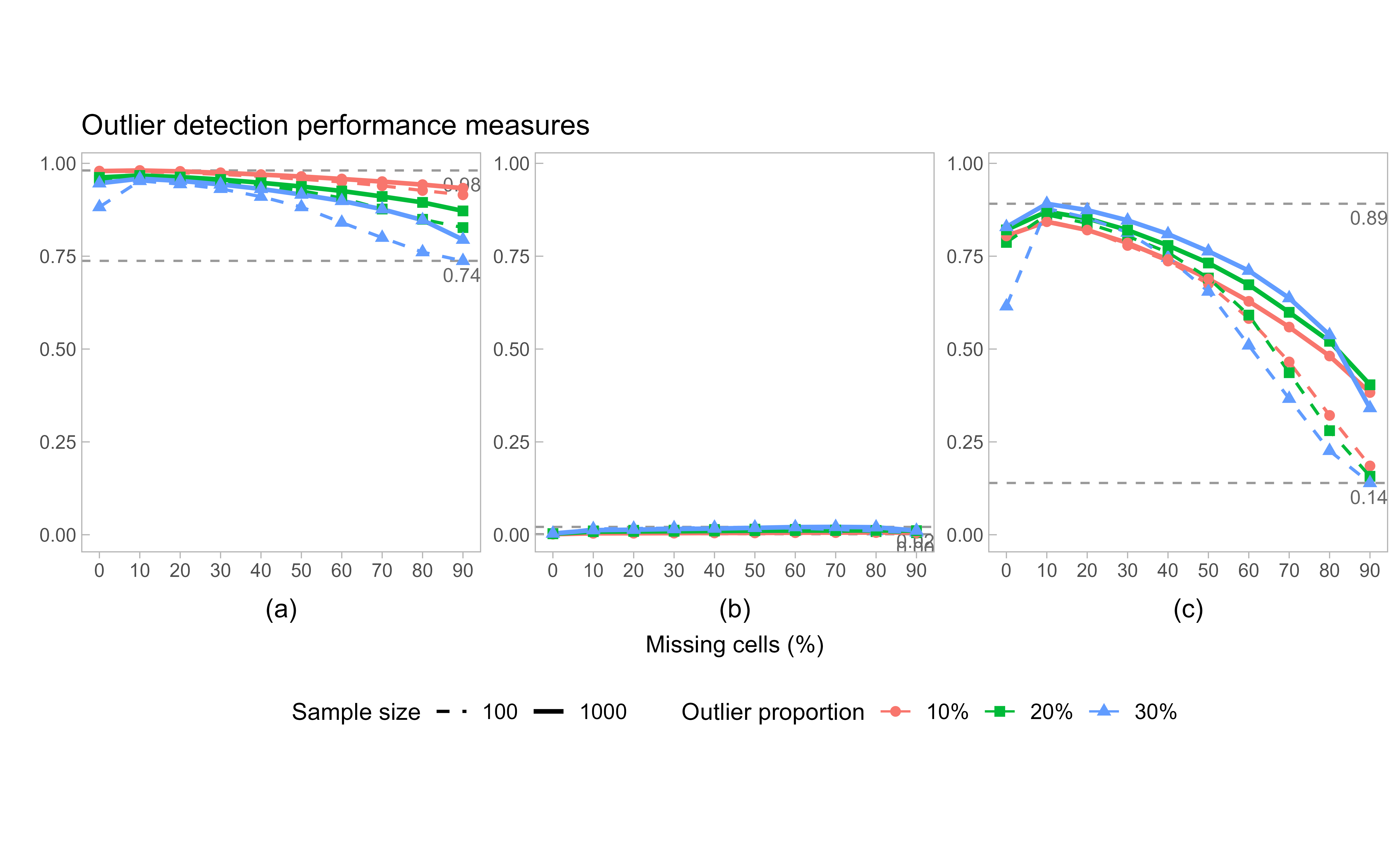}
    \caption{Outlier point detection performance metrics: (a) Average accuracy scores, (b) average FPR, and (c) TPR scores.}
    \label{fig:detection}
\end{figure}

Since the EM-based algorithm can flag outliers, we inspect the performance for outlier detection. Since the number of outliers is fewer than typical points, we report the detection's false positive rate (FPR) and true positive rate (TPR), in addition to the accuracy score. The averages of these metrics are plotted in \figurename~\ref{fig:detection}. The accuracy is, as expected, quite high, even when 90\% of cells are missing. This can be explained by the corresponding decrease in the TPR for high percentages of missing values: that is, there is a bias towards classifying a point as typical amidst sparse data rows, as supported by the low FPR. This outcome is reasonable since a row that is an outlier is more difficult to classify as an outlier when subject to high rates of missingness. A data point is less likely to be imputed as an outlier within the E-step of the EM-based algorithm. Notice though, that sample size does play a greater role in these scenarios. The outlier detection metrics for $n$ = 1000, are better overall compared to the small sample size $n$ = 100. 

To explain the surprisingly concave behaviour around $10\%$ in \figurename~\ref{fig:detection}, we turn to parameter recovery. We consider the bias and root mean squared error (RMSE) of the parameters, defined as:
\begin{align*}
    \mathrm{bias}(\hat{\bm\theta}) = \frac1{1000}\sum_{b=1}^{1000} \sum_{k=1}^{p}\left( \hat{\bm \theta}_b - \bm \theta \right)
    \quad   
    \text{and}
    \quad
    \mathrm{RMSE}(\hat{\bm\theta}) = \sqrt{\frac1{1000}\sum_{b=1}^{1000} \sum_{k=1}^{p}\left( \hat{\bm \theta}_b - \bm \theta \right)^2},
\end{align*}
where $\hat{\bm \theta}_b$ is the estimated parameter for the $b^{th}$ dataset with $b=1,\dots,1000$, and $\bm \theta$ denotes the true parameter.
\begin{figure}[H]
    \centering
    \includegraphics[trim={0cm 0.5cm 0cm 0.7cm}, clip, width=0.7\linewidth]{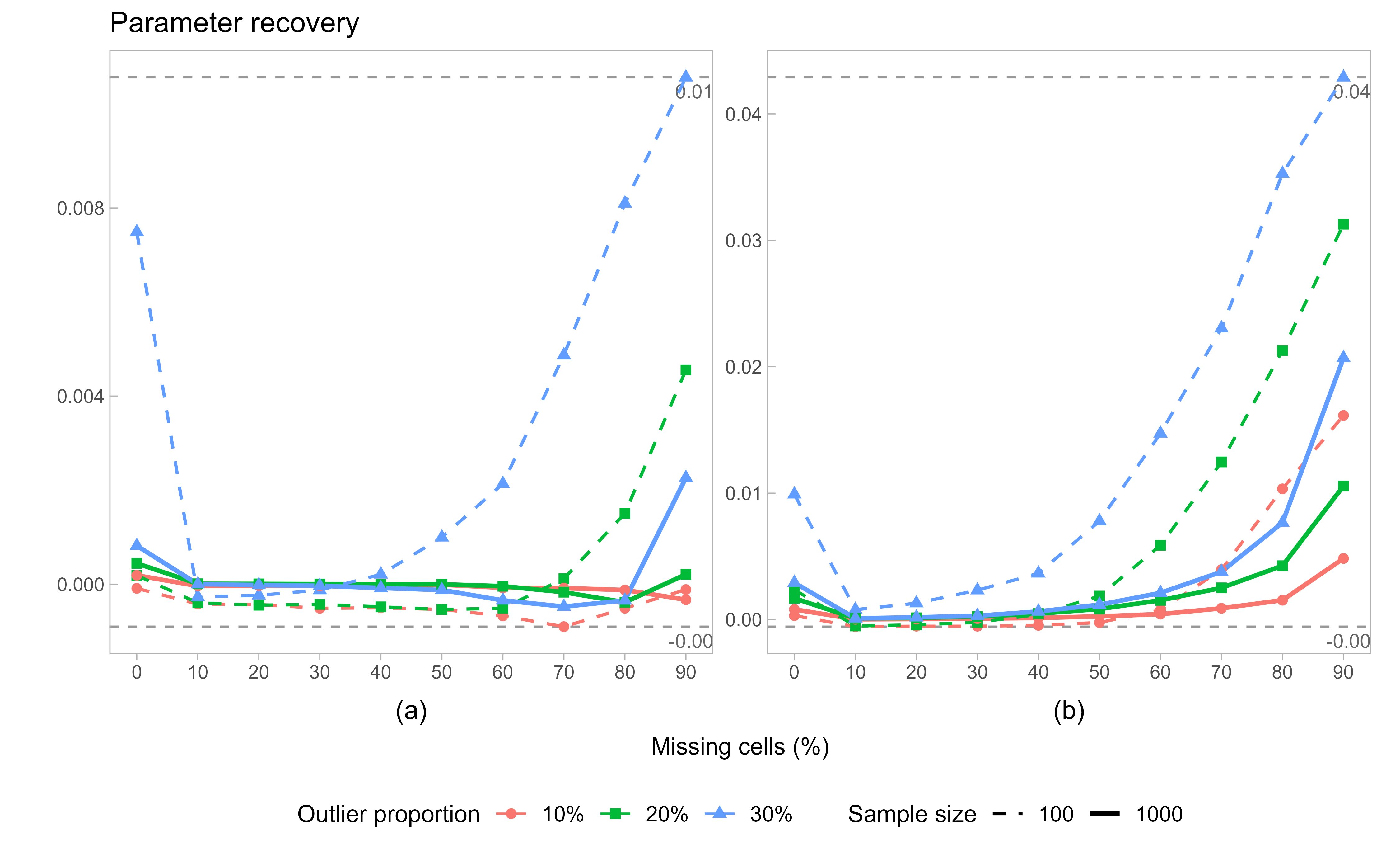}
    \caption{(a) Bias of $\hat\gamma_1$ and (b) bias of $\hat\gamma_2$.}
    \label{fig:gam_bias}
\end{figure}

\begin{figure}[H]
    \centering
    \includegraphics[trim={0cm 0.5cm 0cm 0.7cm}, clip, width=0.7\linewidth]{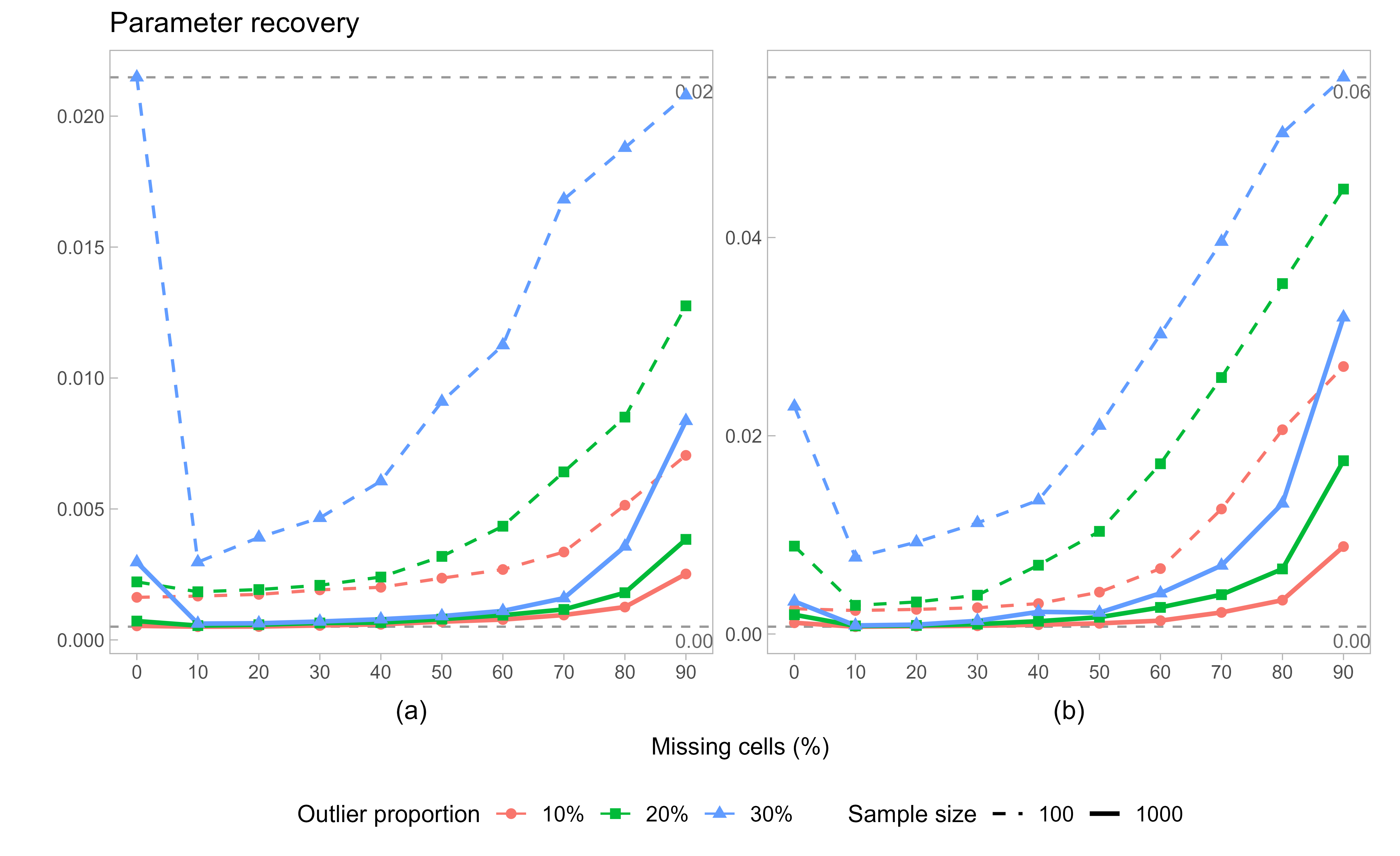}
    \caption{(a) RMSE of $\hat\gamma_1$ and (b) RMSE of $\hat\gamma_2$.}
    \label{fig:gam_rmse}
\end{figure}
\figurename~\ref{fig:gam_bias} and \figurename~\ref{fig:gam_rmse} illustrate a non-monotonic relationship between the proportion of missing data and the recovery of the variability parameter. When the data are fully observed, the parameter tends to be overestimated, particularly as the level of contamination increases, which can be explained by the contaminated model trying to accommodate a separate cluster under a 2-component restriction. However, introducing a moderate percentage of missingness substantially reduces this bias, with the improvement becoming more pronounced as the sample size increases. This suggests that there is a favourable range of missingness percentages within which parameter recovery is most accurate, which is also reflected in the corresponding improvements in clustering and outlier detection. As the proportion of missing data becomes large, the loss of information begins to outweigh this benefit and the bias increases again. Thus, rather than a strictly positive relationship between bias and the missingness percentage as one would expect, the convex-like trend in the results indicates a 'sweet spot' of moderate missingness in which the model achieves particularly favourable parameter recovery while maintaining suitable identification of outliers. Interestingly, \figurename~\ref{fig:mu_bias} demonstrates that the estimation of the means remain unbiased despite the percentage of values missing from a dataset, suggesting that the parametrisation is advantageous against both contamination and missing values.
Of course, its RMSE displays a similar behaviour to what was observed in estimating the variability, as illustrated in \figurename~\ref{fig:mu_rmse}. Particularly, the RMSE of the estimator $\hat{\bm \mu}$ is affected by how much information is missing from a dataset.  
\begin{figure}[H]
    \centering
    \includegraphics[trim={0cm 0.5cm 0cm 0.8cm}, clip, width=0.7\linewidth]{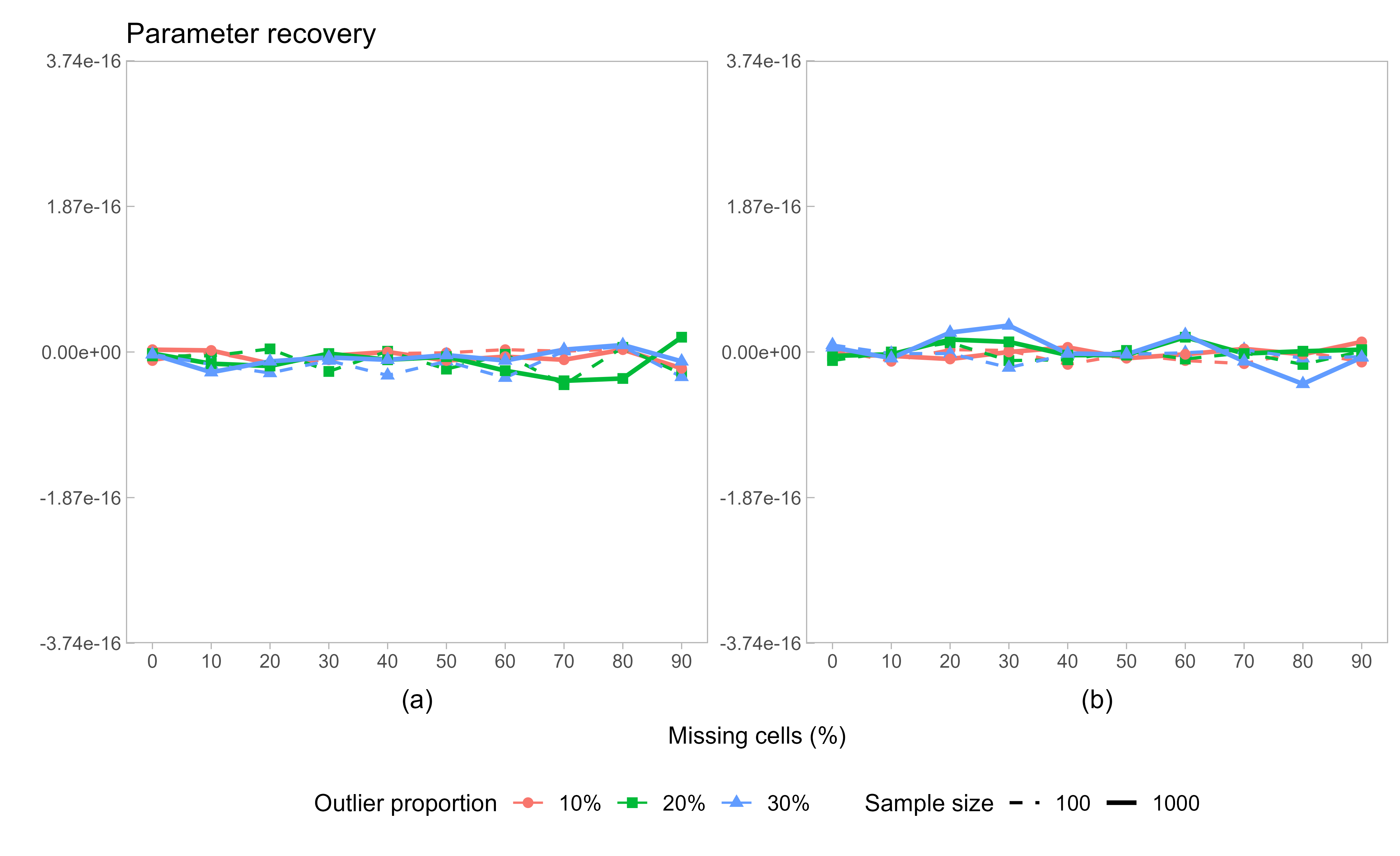}
    \caption{(a) Bias of $\hat {\bm \mu}_1$ and (b) bias of $ \hat{\bm \mu}_2$.}
    \label{fig:mu_bias}
\end{figure}
\begin{figure}[H]
    \centering
    \includegraphics[trim={0cm 0.5cm 0cm 0.7cm}, clip, width=0.7\linewidth]{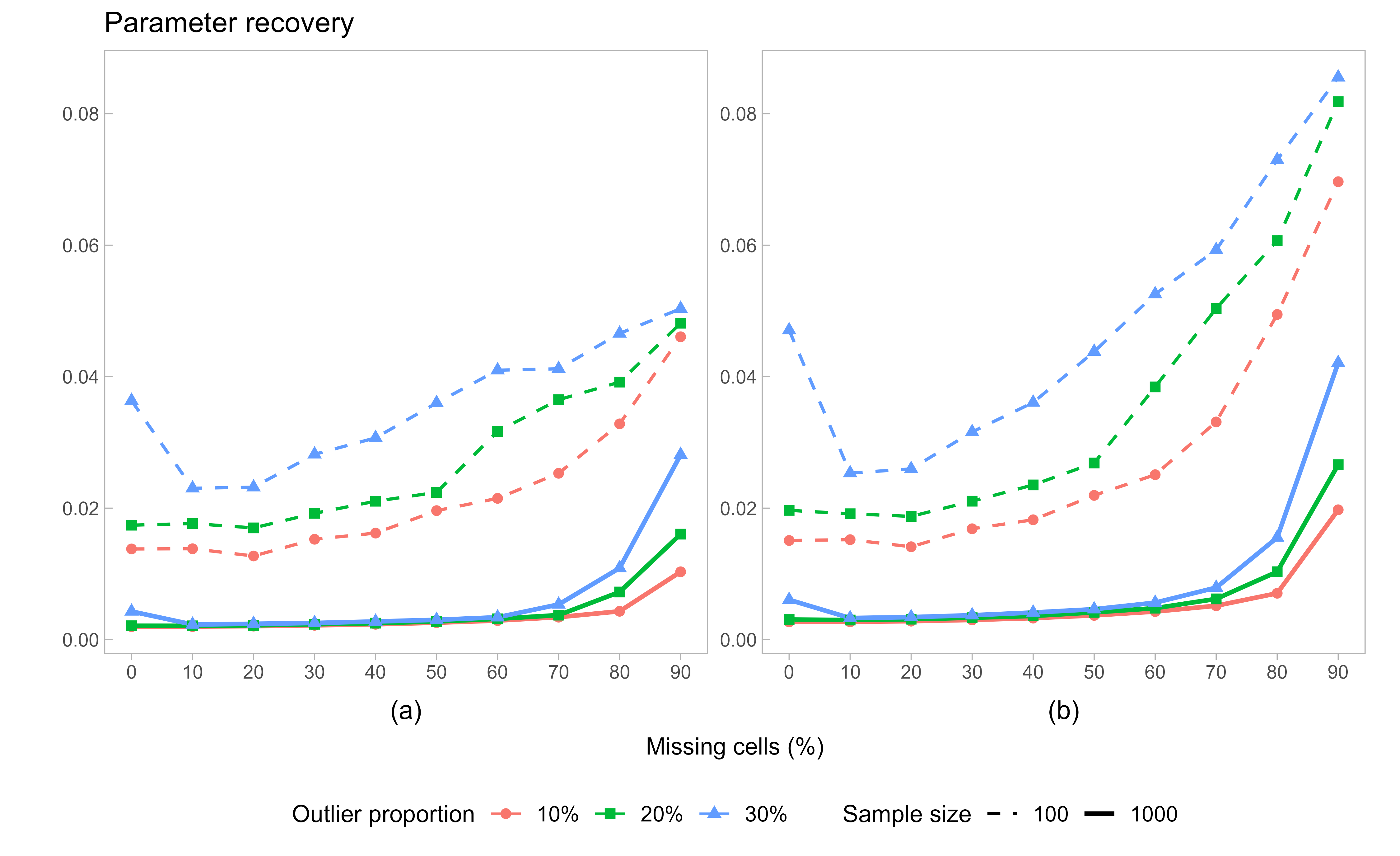}
    \caption{(a) RMSE of $\hat {\bm \mu}_1$ and (b) RMSE of $\hat {\bm \mu}_2$.}
    \label{fig:mu_rmse}
\end{figure}
%\subsection{Model selection}

In practice, the number of components is not known beforehand. The preferred approach in literature is to fit \eqref{cont_dirichlet} on a dataset for a range of pre-determined number of components, and compare their likelihoods through model selection metrics, such as the Akaike Information Criterion (AIC), Bayesian Information Criterion (BIC), and the Integrated Classification Likelihood (ICL), among others. Hence, for all 1000 noisy and potentially incomplete datasets simulated, the contaminated Dirichlet distribution is fitted for components $1\,\dots,p-1$, where $p$ is the dimension of the random vector. The reason for choosing such an upper bound is to maintain model identifiability \citep{identifiability}. We plot the average number of components that optimises the model selection criterion, given in \figurename~\ref{fig:selection}. We compare this with the scenario where finite mixtures of the reference model \eqref{mean_dirichlet} are fitted to the same datasets, and plot the average number of components identified by the model selection metric, in \figurename~\ref{fig:selection_ref}. The AIC, BIC, and ICL are more likely to identify the correct number of components present in the data, when fitting model \eqref{fmm}. For higher percentages of missing values, this average decreases, since sparser datasets produce lower log-likelihood values, thus favouring fewer components to describe the dataset than what is true. Additionally, the average number of components the AIC selects is higher overall compared to BIC and ICL. This is expected as the AIC penalty is weaker compared to the latter criteria. Thus, it tends to favour higher components, overall. This also explains why the average number of components does not drop to 1 for the AIC when datasets have higher rates of missing values.
\begin{figure}[H]
    \centering
    \includegraphics[trim={0.5cm 0.5cm 0cm 2.1cm}, clip, width=0.8\linewidth]{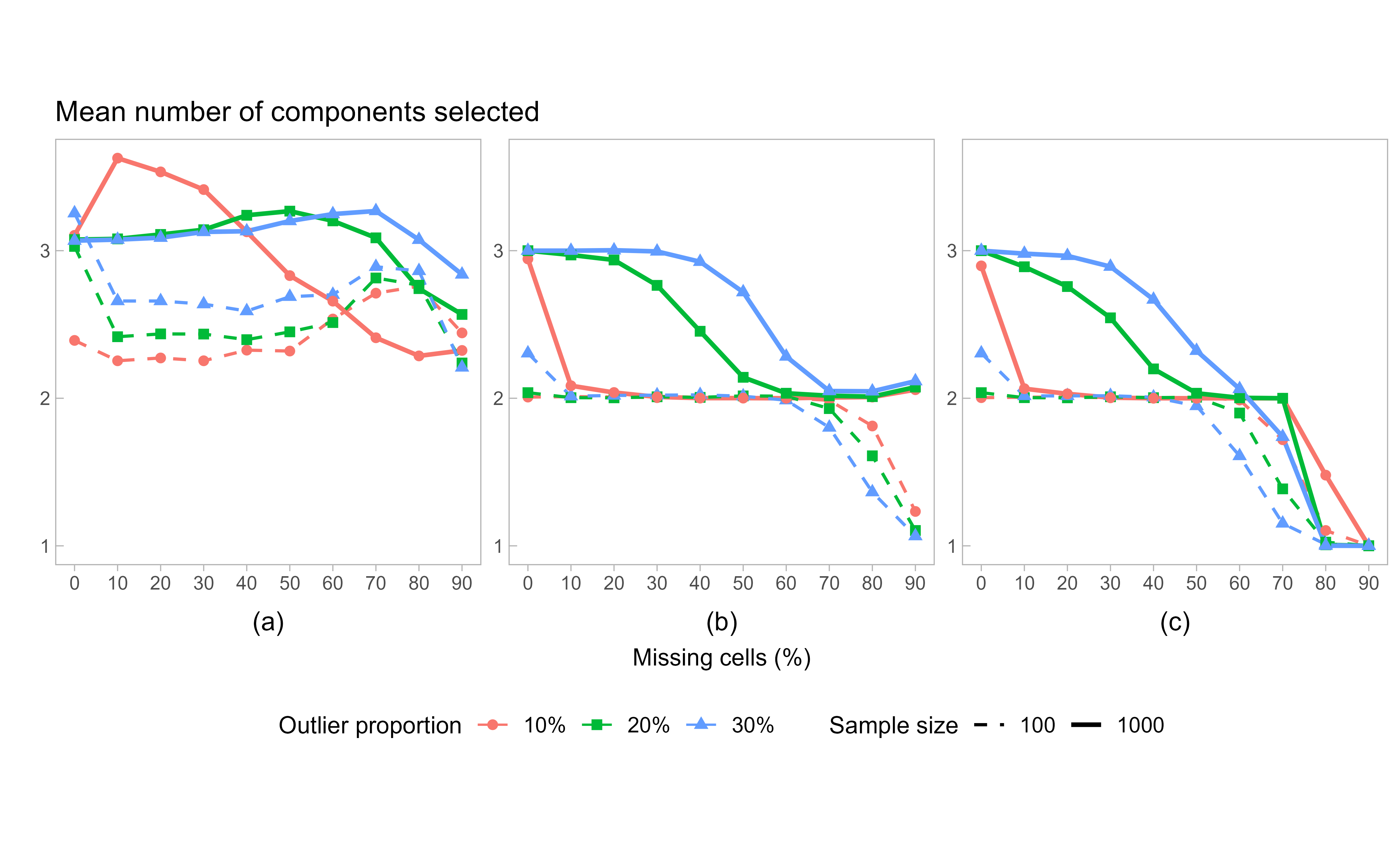}
        \caption{Average of the number of components identified in the dataset under the contaminated model: (a) AIC, (b) BIC, and (c) ICL.}
    %\caption{Probability of model selection metric identifying two components present in the dataset under the contaminated model: (a) AIC, (b) BIC, and (c) ICL.}
    \label{fig:selection}
\end{figure}

\begin{figure}[H]
    \centering
    \includegraphics[trim={0.5cm 0.5cm 0cm 2.1cm}, clip, width=0.8\linewidth]{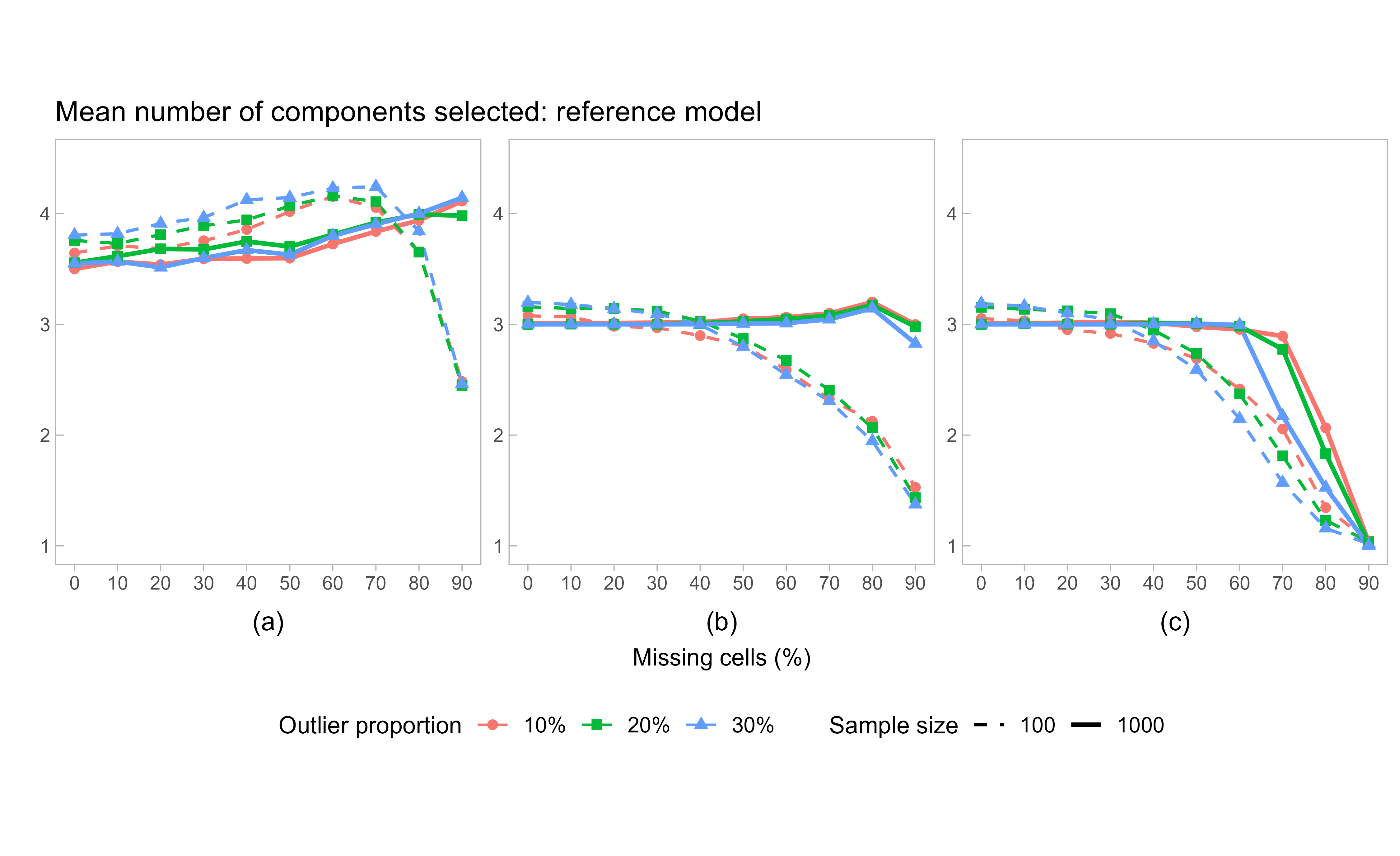}
    \caption{Average of the number of components identified in the dataset under the reference model: (a) AIC, (b) BIC, and (c) ICL.}
    \label{fig:selection_ref}
\end{figure}
What is interesting to note is that for large datasets, noise affects model selection when the missingness percentage is low and the outlier percentage is relatively high. This ties in with uniform noise is distributed as on the simplex, following a $\bD_{\simplex}\left(\frac{1}{p}\bm 1, \frac{1}{p}\right)$ distribution. Thus, for a relatively fuller dataset, the model selection criteria tend to favour a three-component model, in which one component accounts for the outlier points. This scenario is analogous to a finite mixture model that incorporates a uniform component. However, accounting for outliers has shown improvement for moderate to higher percentages of missing values compared with the reference case. Thus, the inference in this paper finds a serendipitous intersection: the challenge of simultaneously identifying outliers in incomplete datasets are best modelled by a finite mixture of contaminated Dirichlet distributions, and fitted with the EM-based algorithm proposed in this paper. 
\section{Application: American Time Use Survey data}
\label{application}
The application demonstrates the usefulness of the algorithm when applying the contaminated Dirichlet distribution and its finite mixtures to gather insights on compositional data, and its effectiveness at identifying clusters within incomplete data. It considers the American Time Use survey data.

The EM-based algorithm is applied to analysing time use data generated by the American Time Use survey (ATUS). The survey is sponsored by the US department of Labor Statistics conducted by the U.S. Census Bureau. The goal of ATUS is to measure how people divide their time among life’s activities and develop nationwide representative estimates of how people spend their time. In ATUS, individuals are randomly selected from a subset of households and respondents are interviewed about how they spent their time on the previous day, including where they were, and whom they were with. Many ATUS users are interested in the amount of time Americans spend doing unpaid, non-market work, which could include unpaid childcare, elder care, housework, and volunteering. The most recent dataset can be downloaded from the US Bureau of Labor Statistics website, and is under the list of Basic files for the year 2025\footnote[1]{U.S. Bureau of Labor Statistics. American Time Use Survey (ATUS). Last Modified Date: 27 July 2026. Internet: \url{https://www.bls.gov/tus/data/datafiles-2025.htm}}. The dataset contains $n=6126$ responses on what activity was conducted and its duration (in minutes). Therefore, each response is compositional in the sense that 24 hours (or 1440 minutes in this case) is a physical constraint that all durations should add up to. Each response is given as a diary entry and is unique to the respondent's lifestyle. The dataset is then pivoted to wider format so that each respondent has the same set of activities they could have potentially done. For this analysis, we focus on the ratio of responsibilities to free time respondents have on a day-to-day basis. This angle has potential for insights into how lifestyle differences between respondents are linked to how their time is utilised, and it ensures that each activity is a reasonable inclusion to a respondent's day. Since a respondent must report their day and only on what they did on the day prior to the survey, activities between responses are inconsistent, thus introducing missing values. A summary of the activities focused on and their missingness percentages are provided in \tablename~ \ref{missing_atus}.

\begin{table}[ht]
    \centering
    \caption{Percentage of missing values by activity.}
    \begin{tabular}{clS}
    \toprule
          & Activity                           &  {Missing (\%)}      \\
            \midrule  
       X1 & Personal care                            &   0.0489716    \\
       X2 & Household responsibilities               &   15.5729677   \\
       X3 & Caring for and helping household members &   79.7747307   \\
       X4 & Work and work-related activities         &   67.6297747   \\
       X5 & Socialising, Relaxing, and Leisure       &   5.2236370    \\
       X6 & Sports, Exercise, and Recreation         &   76.9996735   \\
       X7 & Other activities                         &   98.5634998   \\
        \bottomrule
    \end{tabular}
    \label{missing_atus}
\end{table}
Overall, the dataset contains 49.116\% of missing cells. Finite mixtures of the Dirichlet distribution are fitted for components $G = 1,2,\dots, 6$. We aim to select a model that shows the greatest evidence for clustering the dataset. Larger values of $G$ increase the risk of unrealistically small mixing proportions that may still be justified by the BIC approximation.\citep{icl}. Thus, we consider the integrated classification likelihood (ICL) measure to circumvent this problem \citep{icl_approx}. It also takes into account the ability of the mixture model to give evidence to the resulting clusters. The penalty of the metric is dependent on the uncertainty in the clusters. 
%The ICL approximation is given as \citep{icl_approx}:
%\begin{align}
%\label{icl}
%    ICL \approx 2l(\widehat{\balpha};\bm{z},\mathcal{X}_o) - P\ln n,
%\end{align}
%where $\widehat{\balpha}$ and $\bm{z}$ are the estimates and MAP probabilities \eqref{responsibilities} respectively, after the proposed algorithm has converged, and $P$ is the number of free parameters in the mixture model that are estimated. 
\begin{figure}
    \centering
    \includegraphics[trim={0cm 0.5cm 0cm 1cm}, clip, width=0.4\linewidth]{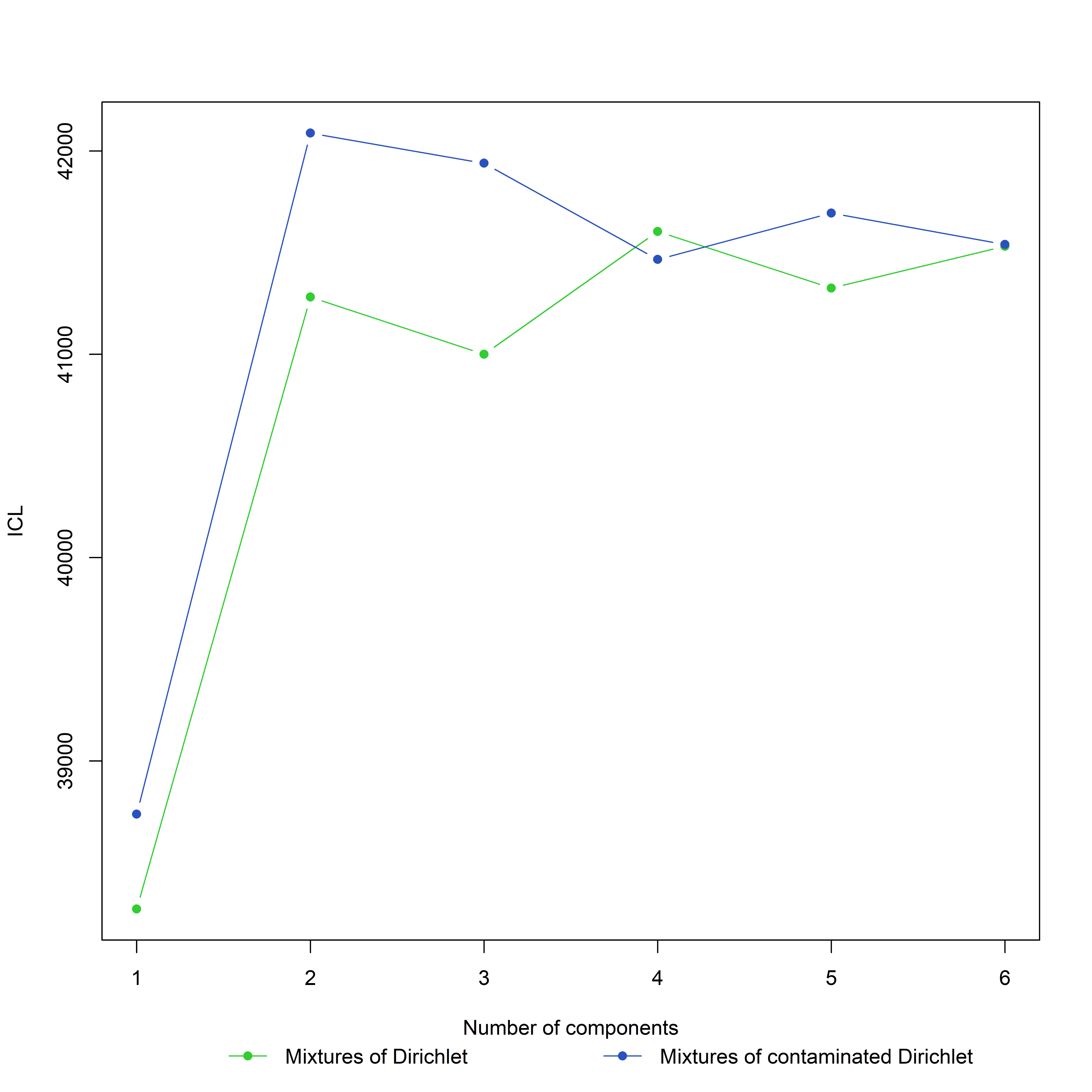}
    \caption{ICL values for predetermined number of components $G=1,\dots,6$. Here, selection is made based on the number of components that maximises the ICL.}
    \label{fig:bic_icl}
\end{figure}

Out of the candidate components that were fitted, a four-component reference Dirichlet mixture model produced the largest ICL value, while the ICL is largest at a two-component contaminated Dirichlet mixture model. Notice that the cluster coloured in pink in \figurename~\ref{fig::ref_fit} by the mixtures of Dirichlet distribution has been identified as outliers indicated by the grey points in \figurename~\ref{fig::ref_fit} when a two-component contaminated Dirichlet mixture model is fitted. Further, the green cluster in \figurename~\ref{fig::ref_fit} has been absorbed into the blue cluster in \figurename~\ref{fig::cont_fit}, suggesting that the variability in the contaminated Dirichlet can better explain these points under one cluster, compared to its reference counterpart. The estimates from fitting a two-component contaminated Dirichlet mixture model are given in \tablename~\ref{table::cont_estimates}. Overall, the model identified 16.389\% of rows as outliers and can be summarised as follows:

Cluster 1 (32\%)- a work-centred group, characterised by high proportions of time devoted to work and personal care, with relatively little leisure or household activity.\\
Cluster 2 (68\%)- a home and leisure-centred group, characterised by substantially less work and more leisure and household responsibilities.

Within each cluster, between 20\% to 30\% of individuals exhibit atypical daily time-use patterns. Rather than forming a separate cluster, these are accommodated through the contamination component, allowing the model to remain flexible to unusual observations while preserving the main behaviour suggested by the identified clusters.
%\begin{figure}[ht]
%    \centering
%    \fbox{\includegraphics[width=0.5\linewidth]{icl6.png}}
%    \caption{Pair-plots of the variables in the ATUS dataset clustered according to the MAP allocations after fitting a $K=6$-component Dirichlet mixture model.}
%    \label{icl6}
%\end{figure}

\begin{figure}[H]
    \centering
    \includegraphics[trim={0.5cm 0.5cm 0.5cm 0.7cm}, clip, width=0.63\linewidth]{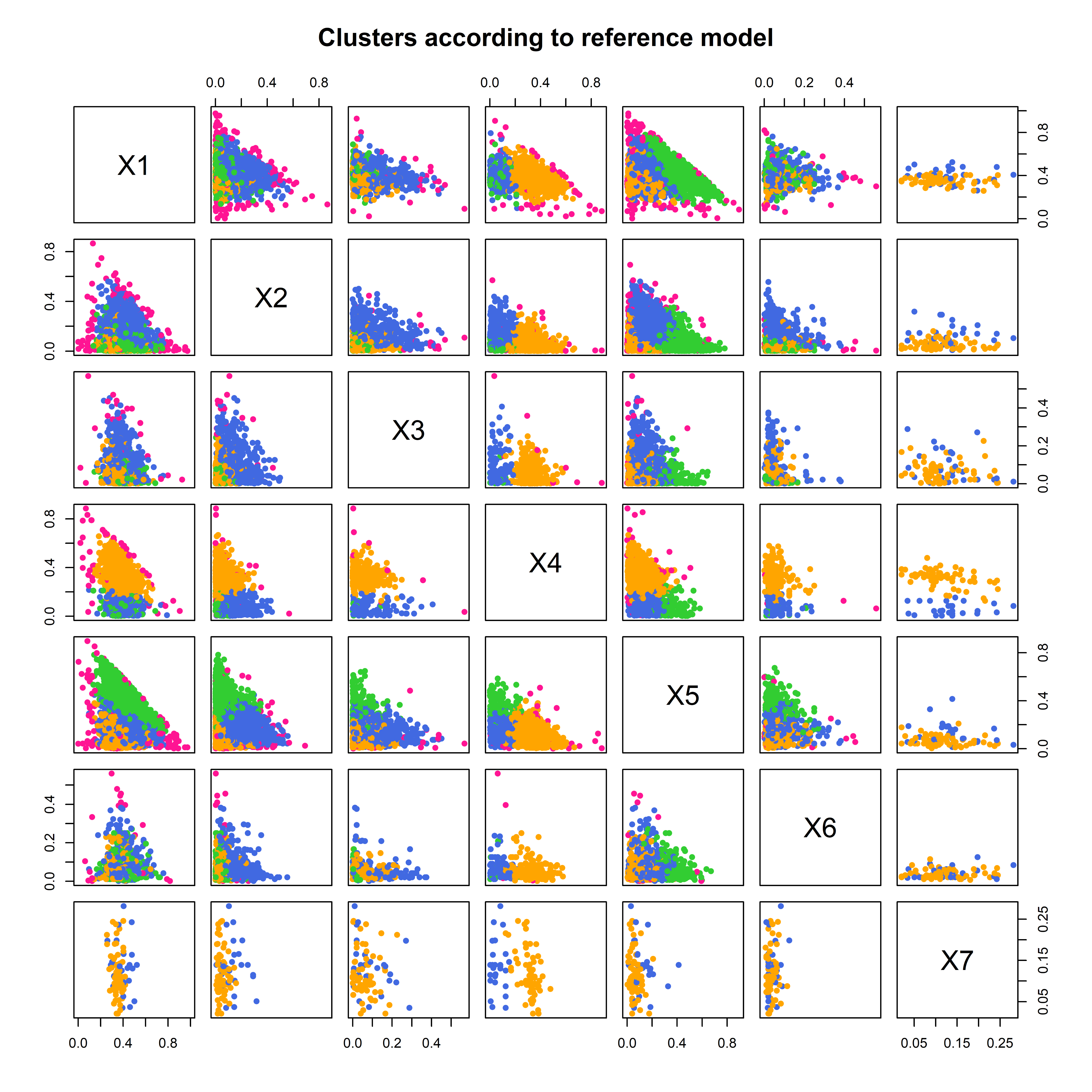}
    \caption{Cluster-wise pair plots of observed time use proportions from the ATUS dataset clustered according to the MAP classifications after fitting a $G=4$-component reference Dirichlet mixture model.}
    \label{fig::ref_fit}
\end{figure}

\begin{figure}[H]
    \centering
    \includegraphics[trim={0.5cm 0.5cm 0.5cm 0.7cm}, clip, width=0.63\linewidth]{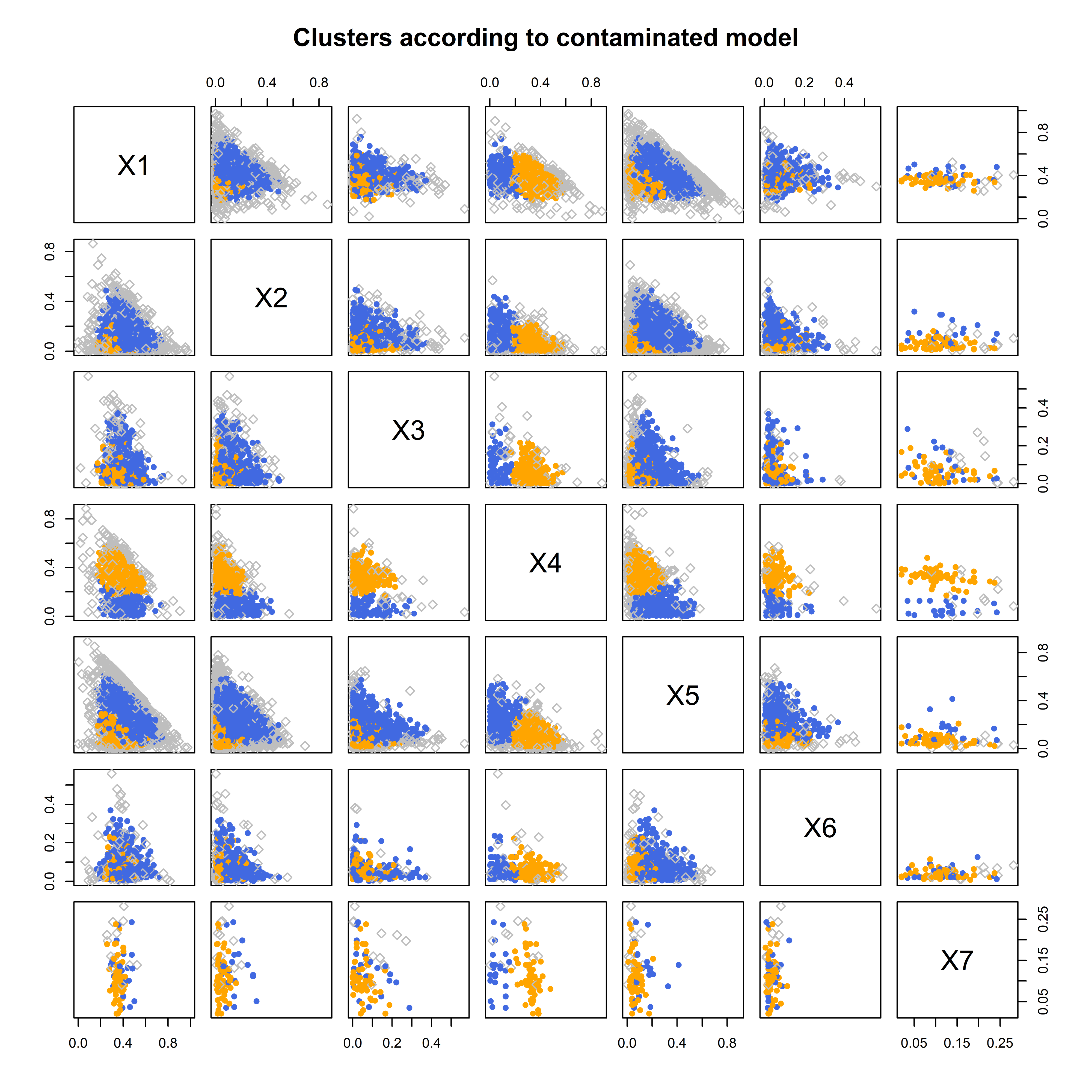}
    \caption{Cluster-wise pair plots of observed time use proportions from the ATUS dataset clustered according to the MAP classifications after fitting a $G=2$-component contaminated Dirichlet mixture model. The grey points are rows that have been identified as outliers by the EM-based algorithm.}
    \label{fig::cont_fit}
\end{figure}

\begin{table}[H]
\centering
\caption{Parameter estimates after fitting a two-component contaminated Dirichlet mixture model.}
\label{table::cont_estimates}
\begin{tabular}{lSSSSSSSSSSS}
\toprule
Cluster &
{$\hat{\pi}_g$} &
{$\hat{\varepsilon}_g$} &
{$\hat{\mu}_{1g}$} &
{$\hat{\mu}_{2g}$} &
{$\hat{\mu}_{3g}$} &
{$\hat{\mu}_{4g}$} &
{$\hat{\mu}_{5g}$} &
{$\hat{\mu}_{6g}$} &
{$\hat{\mu}_{7g}$} &
{$\hat \gamma_g$}     &
{$\hat \eta_g$}       \\
\midrule
$g=1$ &
0.3170296  & 
0.2153885  & 
0.38615297 &
0.05349394 &
0.04233250 &
0.32891516 &
0.10700475 &
0.04713015 &
0.03497054 &
0.03435703 &
2.611211  \\

$g=2$ &
 0.6829704 &
 0.2782382 &
0.43247635 &
0.12197052 &
0.05558154 &
0.05915569 &
0.25809062 &
0.05841992 &
0.01430536 &
0.06677725 &
2.990318   \\
\bottomrule
\end{tabular}
\end{table}

\section{Conclusion}
The methodology has established that the maximum likelihood estimates of observed log-likelihood of a mean-parametrised contaminated Dirichlet distribution when data are missing-at-random exist, are unique, and are identifiable up to permutation provided that the number of components does not exceed the dimension of the simplex. It has also established an idea of what outliers may look like on the simplex, through uniform noise. The EM-based algorithm developed simultaneously assigns each data point a soft cluster membership and an outlier belonging, while accounting for incompleteness. Simulations demonstrate clustering and outlier detection capabilities, even under severe data incompleteness, including observations with only a single observed component and datasets containing up to 90\% of unobserved values.\\
The proposed methodology was also applied to a real dataset involving compositional variables subject to physical additive constraints. The algorithm successfully estimated finite mixture models and identified meaningful clusters directly on the simplex. Because the analysis remains on the simplex, the resulting clusters can be interpreted directly in terms of the original compositional variables, without requiring transformations or ad hoc adjustments.
More broadly, this work reinforces the role of mixtures of the Dirichlet distribution as suitable models for compositional data. When incomplete compositions arise, analysts have historically relied on methods designed for unconstrained data, adapting the data rather than the model. The framework proposed here reverses this perspective by developing estimation tools that respect the geometry and constraints of compositional data. By enabling likelihood-based inference and clustering directly on the simplex in the presence of unobserved parts, this work opens the door to further developments in simplex-based modeling, including richer mixture models and more flexible inference procedures.

\begin{appendix}
\label{appn}
\section{Proofs of supplementary results}
\begin{proof}[Proof of lemma \ref{coercive_lemma} ]
    \label{coercive_proof}
    Let $ \theta = \| [\gamma, \eta]^{\top} \|_2$. Then we evaluate the limit:
    \begin{align*}
        \underset{\theta \rightarrow \infty}{\lim} \ell_o( \bpsi; \mathcal{X}_o) 
        &= \underset{\theta \rightarrow \infty}{\lim}\sum_{i=1}^n \ln f_{\mathcal{CD}}(\bx_{i,o}^v;\bmu_{i,o}^v, \gamma, \eta, \varepsilon)\\
        &= \sum_{i=1}^n \ln \left[ \varepsilon \underset{\theta \rightarrow \infty}{\lim}f_{\mathcal{D}}(\bx_{i,o}^v;\bmu_{i,o}^v, \eta\gamma) + (1-\varepsilon)\underset{\theta \rightarrow \infty}{\lim}f_{\mathcal{D}}(\bx_{i,o}^v;\bmu_{i,o}^v, \gamma) \right] 
    \end{align*}
    First note $\forall~\eta, \gamma>1$ we have that $\eta \gamma = \sqrt{\eta^2 \gamma^2} \geq \| [\gamma, \eta]^{\top} \|$.
    Further, $0 < \frac{1}{\eta \gamma}< \frac{1}{\gamma} < 1,$ and 
    \[
    f_{\mathcal{CD}}(\bx_{i,o}^v;\bmu_{i,o}^v, \gamma, \eta, \varepsilon) < f_{\mathcal{D}}(\bx_{i,o}^v;\bmu_{i,o}^v, \eta\gamma) + f_{\mathcal{D}}(\bx_{i,o}^v;\bmu_{i,o}^v, \gamma).\]
    
    Notice that $\mu_k<1 ~\forall k \in \os_i$ and $\forall i=1,\dots,n$, implying that $ \frac{\mu_k}{\eta \gamma} < \frac{\mu_k}{\gamma}$. From the monotonic decrease of $\Gamma(\cdot)$ on $(0,1]$ we have that
    \begin{align}
    \label{gamma_inequality}
        \frac{\Gamma\left(\frac{1}{\gamma}\right)}{\displaystyle \prod_{k \in \os_i}^{p_o} \Gamma\left(\frac{\mu_k}{\gamma} \right) } < \frac{\Gamma\left(\frac{1}{\eta \gamma}\right)}{\displaystyle \prod_{k \in \os_i}^{p_o} \Gamma\left(\frac{\mu_k}{\eta \gamma} \right) }.
    \end{align}
    Further, $\bx_i \in \simplex$ , which implies that
    \begin{align}
    \label{x_inequality}
        \prod_{k \in \os_i}^{p_o} x_{ik}^{\frac{\mu_k}{\gamma}-1} < \prod_{k \in \os_i}^{p_o} x_{ik}^{\frac{\mu_k}{\eta\gamma}-1}
    \end{align}
    Combining resutls \eqref{gamma_inequality} and \eqref{x_inequality}, we conclude that $\forall~\eta, \gamma >1:$
    \begin{align*}
         f_{\mathcal{CD}}(\bx_{i,o}^v;\bmu_{i,o}^v, \gamma, \eta, \varepsilon) < 2  f_{\mathcal{D}}(\bx_{i,o}^v;\bmu_{i,o}^v, \eta \gamma).
    \end{align*}
    Without loss of generality, suppose $\mu_{p_o} = \min \{ \mu_k:k \in \os_i\}$. Since $\Gamma(z)>1 ~ \forall z \in (0,1]$. Since $\Gamma\left( \frac{\mu_k}{\eta\gamma} \right) > \Gamma\left(\frac{1}{\eta\gamma}\right)~ \forall k \in \os_i$ we find that $ \forall ~\gamma, \eta >1$
    \begin{align*}
        \Gamma\left(\frac{1}{\eta\gamma} \right) < \displaystyle\prod_{k \in \os_i}^{p_o -1 } \Gamma\left( \frac{\mu_k}{\eta \gamma} \right) 
        \implies
       0
       <\frac{\Gamma\left( \frac{1}{\eta\gamma} \right)}{ \displaystyle\prod_{k \in \os_i}^{p_o} \Gamma\left(\frac{\mu_k}{\eta \gamma} \right) } 
        <
        \frac{1}{\Gamma \left(\frac{\mu_{p_o}}{\eta \gamma} \right)}
%        < 
%        \frac{1}{\Gamma \left(\frac{\mu_{p_o}}{\gamma} \right)}
        < \frac{1}{\Gamma \left(\frac{\epsilon}{\gamma} \right)},
    \end{align*}
Further, we find that $\frac{1}{\Gamma \left(\frac{\epsilon}{ \gamma} \right)} \rightarrow 0$ as $\gamma \rightarrow \infty$ through $\theta \rightarrow \infty$. By the squeeze theorem, 
\begin{align*}
%\label{norm_const_lim}
\underset{\theta \rightarrow \infty}{\lim}\frac{\Gamma\left( \frac{1}{\eta\gamma} \right)}{\displaystyle \prod_{k \in \os_i}^{p_o} \Gamma\left(\frac{\mu_k}{\eta \gamma} \right) }  = 0.    
\end{align*}
Now as $\gamma \rightarrow \infty$ through $\theta$,
\begin{align*}
    \lim_{\theta\rightarrow\infty} \prod_{k\in O_i}^{p_0} x_{ik}^{\frac{\mu_k}{\gamma}-1} \left( 1-\|\bx_{i}\|_2 \right)^{ \frac{1-\|\boldsymbol{\mu}_{i,o}\|}{\gamma}-1 }
=\lim_{\theta\rightarrow\infty}& 
\frac{\left( \displaystyle\prod_{k\in O_i}^{p_0} x_{ik}^{\mu_k} \left( 1-\| \bx_{i} \|_2\right)^{1-\|\bmu_{i,o}\|_2}\right)^{ \frac{1}{\gamma}} }{\displaystyle\prod_{k\in O_i}^{p_0} x_{ik}^{\mu_k} \left( 1-\| \bx_{i} \|_2 \right)^{1-\|\bmu_{i,o}\|_2}}.
\end{align*}
Since $
0<
\prod_{k\in O_i}^{p_0}
x_{ik}^{\mu_k}
\left(
1-\|\bx_{i}\|_2
\right)^{1-\|\bmu_{i,o}\|_2}
<1,
$
we obtain
\[
\lim_{\theta\rightarrow\infty}
\left(
\prod_{k\in O_i}^{p_0}
x_{ik}^{\mu_k}
\left(
1-\|\bx_{i}\|_2
\right)^{1-\|\bmu_{i,o}\|_2}
\right)^{\frac1\gamma}
=1.
\]

Therefore
\[
\lim_{\theta\rightarrow\infty}
f_{\mathcal D} \left( \bx_{i,o}^v;\bmu_{i,o}^v,\eta\gamma \right) =0
\cdot
\left(\prod_{k\in \os_i}^{p_0} x_{ik}^{\mu_k} \left( 1- \|\bx_{i}\| _2\right)^{1-\|\bmu_{i,o}\|_2} \right)^{-1} =0,
\]
pointwise on $\mathcal X_{o}$. This is for arbitrary $\boldsymbol{\mu}\in \simplex^{\epsilon}$. That is,
\[
\lim_{\theta\rightarrow\infty}
\sup_{\bmu \in \simplex^{\epsilon} }
f_{\mathcal D} \left( \bx_{i};\bmu_{i,o}, \eta\gamma \right) = 0 \]
by the squeeze theorem. This implies that 
\[ \lim_{\theta\rightarrow\infty} \ln f_{\mathcal {CD}} \left( \bx_{i};\bmu_{i,o},\gamma, \eta, \varepsilon \right) = -\infty, \] as required.
\end{proof}

\begin{proof}[Proof of Lemma \ref{neg_def_lemma}]
    \label{neg_def_proof}
    Simplifying, $\bm H_{i,o}$ can be written in matrix form:
    \[
    \bm H_{i,o} = \bm D + \bm\nabla_{\bmu}f_{\mathcal D}~ \bm \xi^{\top}, \]
    where
    \[
    \bm D = \operatorname{diag} \left( -\frac{f(\bx_i^v;\bmu_{i,o}^v,\gamma)} \gamma  \varphi^{(1)} \!\left( \frac{\mu_k}{\gamma} \right)\right)_{k\in\os_i}
    \]
    and
    \[
    \bm{\xi} = - \left( \varphi \!\left(\frac{\mu_k}{\gamma}\right) -\frac1 \gamma\ln x_{ik} \right)_{k\in\os_i}.
    \]
    Now $\bm D \prec \bm0,$ and $ \bm \nabla_{\bmu}f\,\bm \xi^\top$ is symmetric and rank one. This implies that 
    \begin{align*}
        \bm \xi = C\nabla_{\bmu}f_{\mathcal D} \text{, where } C=f_{\mathcal D}(\bx_{i,o}^v;\bmu_{i,o}^v,\gamma).
    \end{align*}
    Hence $\bm H_{i,o}$ is perturbed by a rank-one matrix, implying at most one eigenvalue is perturbed, leaving the rest undisturbed. Therefore
    \begin{align*}
        |\bm H_{i,o}| = \left( 1+ \bm \nabla_{\bmu}f_{\mathcal D}^\top \bm D^{-1} \bm \nabla_{\bmu}f_{\mathcal D} \right) |\bm D|.
    \end{align*}
    The perturbation can change the parity of the number of negative eigenvalues. However,
    \[
    1+ \bm \nabla_{\bmu}f_{\mathcal D}^\top \bm D^{-1} \bm \nabla_{\bmu}f_{\mathcal D} >0, \qquad \forall\; \bmu\in\simplex^\epsilon.\]
    Hence
    \[ \operatorname{sign}(\bm H_{i,o}) = \operatorname{sign}(\bm D) \implies \bm H_{i,o} \prec \bm 0.\]
\end{proof}
\begin{proof}[Proof $\tilde{Q}$ is locally Lipschitz continuous on $\bmu_g$]
    For any subset
\[
\mathcal B=\left\{\bx\in\mathbb{R}_+^p:\min(\bx)\geq\varepsilon\right\},
\]
where $\varepsilon>0$ is fixed, for any $\bx,\by\in\mathcal B$. Let $\delta(x,y) = \varphi^{(1)}(x) - \varphi^{(1)}(y) $
\begin{align*}
\left\|\bm H_{g}(\bx)-\bm H_{g}(\by)\right\|_2
&\leq
\sup_{\ell\in \{1,\dots,p\} } 
\left| -\frac{n_{\mathrm{ref.},g}}{\gamma_g^2} \delta\left(\frac{y_\ell}{\gamma_g},\frac{x_\ell}{\gamma_g}\right)
+\frac{n_{\mathrm{cont.},g}}{(\eta_g\gamma_g)^2}  
\delta\left(\frac{y_\ell}{\eta_g\gamma_g},\frac{x_\ell}{\eta_g\gamma_g}\right)
\right| \\
&\leq
\sup_{\ell\in \{1,\dots,p\} } 
\left| \frac{n_{\mathrm{ref.},g}}{\gamma_g^2} \delta\left(\frac{y_\ell}{\gamma_g},\frac{x_\ell}{\gamma_g} \right) \right|
+\sup_{\ell\in \{1,\dots,p\} }\left| \frac{n_{\mathrm{cont.},g}}{(\eta_g\gamma_g)^2}  \delta\left(\frac{y_\ell}{\eta_g\gamma_g},\frac{x_\ell}{\eta_g\gamma_g}\right) \right| \\
&\leq c\left| \delta(\zeta y_{sup}, \zeta x_{sup}) \right|
\end{align*}

Now, by the Mean Value Theorem
\begin{align*}
\left|\varphi^{(1)}(\zeta x_{sup}) - \varphi^{(1)}(\zeta y_{sup}) \right|
&= c_1 \left| \varphi^{(2)}(z_{sup}) \right| \zeta |x_{sup}-y_{sup}|,\\
\text{for some $z_{sup}\in[\zeta x_{sup},\zeta y_{sup}]$ and some $c_1>0$. Hence,}\\
\sup_{z_{sup}\in[\zeta x_{sup},\zeta y_{sup}]} c_1 \left| \varphi^{(2)}(z_{sup}) \right| \zeta |x_{sup}-y_{sup}| &\leq L\|\bx-\by\|_2,
\end{align*}
for some $L\in(0,\infty)$.
\end{proof}
 %% if no title is needed, leave empty \section*{}.
%Appendices should be provided in \verb|{appendix}| environment,
%before Acknowledgements.
%
%If there is only one appendix,
%then please refer to it in text as \ldots\ in the \hyperref[appn]{Appendix}.
\end{appendix}
%%%%%%%%%%%%%%%%%%%%%%%%%%%%%%%%%%%%%%%%%%%%%%
%% Example with multiple Appendixes:        %%
%%%%%%%%%%%%%%%%%%%%%%%%%%%%%%%%%%%%%%%%%%%%%%
%\begin{appendix}
%\section{Title of the first appendix}\label{appA}
%If there are more than one appendix, then please refer to it
%as \ldots\ in Appendix \ref{appA}, Appendix \ref{appB}, etc.
%
%\section{Title of the second appendix}\label{appB}
%\subsection{First subsection of Appendix \protect\ref{appB}}
%
%Use the standard \LaTeX\ commands for headings in \verb|{appendix}|.
%Headings and other objects will be numbered automatically.
%\begin{equation}
%\mathcal{P}=(j_{k,1},j_{k,2},\dots,j_{k,m(k)}). \label{path}
%\end{equation}
%
%Sample of cross-reference to the formula (\ref{path}) in Appendix \ref{appB}.
%\end{appendix}

%%%%%%%%%%%%%%%%%%%%%%%%%%%%%%%%%%%%%%%%%%%%%%
%% Acknowledgements                         %%
%% should be provided in the                %%
%% Acknowledgements section.                %%
%%%%%%%%%%%%%%%%%%%%%%%%%%%%%%%%%%%%%%%%%%%%%%
\begin{acks}[Acknowledgments]
The authors would like to thank the anonymous referees, an Associate
Editor and the Editor for their constructive comments that improved the
quality of this paper.
\end{acks}

%%%%%%%%%%%%%%%%%%%%%%%%%%%%%%%%%%%%%%%%%%%%%%
%% Funding information, if any,             %%
%% should be provided in the                %%
%% funding section.                         %%
%%%%%%%%%%%%%%%%%%%%%%%%%%%%%%%%%%%%%%%%%%%%%%
\begin{funding}
%The first author was supported by NSF Grant DMS-??-??????.
%The second author was supported in part by NIH Grant ???????????.
\end{funding}

\bibliographystyle{imsart-number} % Style BST file (imsart-number.bst or imsart-nameyear.bst)
\bibliography{bibliography}       % Bibliography file (usually '*.bib')

%% or include bibliography directly:
%\begin{thebibliography}{9}
%
%\bibitem{r1}
%\textsc{Billingsley, P.} (1999). \textit{Convergence of
%Probability Measures}, 2nd ed.
%Wiley, New York.
%\MR{1700749}
%
%\bibitem{r2}
%\textsc{Bourbaki, N.}  (1966). \textit{General Topology}  \textbf{1}.
%Addison--Wesley, Reading, MA.
%
%\bibitem{r3}
%\textsc{Ethier, S. N.} and \textsc{Kurtz, T. G.} (1985).
%\textit{Markov Processes: Characterization and Convergence}.
%Wiley, New York.
%\MR{838085}
%
%\bibitem{r4}
%\textsc{Prokhorov, Yu.} (1956).
%Convergence of random processes and limit theorems in probability
%theory. \textit{Theory  Probab.  Appl.}
%\textbf{1} 157--214.
%\MR{84896}
%\end{thebibliography}

\end{document}